\documentclass[a4paper,journal,twoside]{IEEEtran}
\IEEEoverridecommandlockouts

\newcommand\hmmax{0}
\newcommand\bmmax{0}

\usepackage{graphicx}
\usepackage{tikz}
\usepackage{tikz-cd}
\usepackage{pgfplots}
\usepackage{xcolor}
\usepackage{color, colortbl}
\usepackage{amsmath}
\usepackage{bbm}
\usepackage{amsfonts, amssymb}
\usepackage{mathtools}
\usepackage{cite}
\usepackage[nolist]{acronym}
\usepackage{amsmath}

\usepackage{booktabs} 		
\usepackage{diagbox}		
\usepackage{upgreek}
\usepackage{bbm}
\usepackage{Files/bm}
\usepackage{Files/winsnotation}
\usepackage{Files/Symbols_OPL}
\usepackage{Files/LVcolours} 

\usepackage[ruled,vlined,noresetcount]{algorithm2e}
\usepackage{setspace}	 	
\usepackage{comment}
\usepackage{physics}
\usepackage{multirow}

\usepackage{array}
\newcolumntype{C}[1]{>{\centering\arraybackslash}p{#1}}

\makeatletter
\def\endthebibliography{%
  \def\@noitemerr{\@latex@warning{Empty `thebibliography' environment}}%
  \endlist
}
\makeatother
\usepackage{amsthm}
\theoremstyle{definition}

\newtheorem{theorem}{Theorem}
\newtheorem{example}{Example}
\newtheorem{lemma}{Lemma}
\newtheorem{corollary}{Corollary}
\DeclareMathOperator{\image}{im}

\newcommand{\px}{p_{\mathrm{X}}}
\newcommand{\py}{p_{\mathrm{Y}}}
\newcommand{\pz}{p_{\mathrm{Z}}}

\pgfplotsset{compat=1.17}

\begin{document}

\title{Cylindrical and M\"{o}bius Quantum Codes \\ for  Asymmetric Pauli Errors}



\author{
    Lorenzo~Valentini,~\IEEEmembership{Member,~IEEE,} Diego Forlivesi,~\IEEEmembership{Graduate~Student~Member,~IEEE,}
    and~Marco~Chiani,~\IEEEmembership{Fellow,~IEEE}
    \thanks{
    The authors are with the Department of Electrical, Electronic, and Information Engineering ``Guglielmo Marconi'' and CNIT/WiLab, University of Bologna, 40136 Bologna, Italy. E-mail: \{lorenzo.valentini13, diego.forlivesi2, marco.chiani\}@unibo.it.
    This work is supported in part by PNRR MUR project PE0000023-NQSTI (Italy) financed by the European Union – Next Generation EU.}
}

\maketitle 

\begin{acronym}
\small
\acro{AWGN}{additive white Gaussian noise}
\acro{BCH}{Bose–Chaudhuri–Hocquenghem}
\acro{CDF}{cumulative distribution function}
\acro{CRC}{cyclic redundancy code}
\acro{LDPC}{low-density parity-check}
\acro{ML}{maximum likelihood}
\acro{MWPM}{minimum weight perfect matching}
\acro{QECC}{quantum error correcting code}
\acro{PDF}{probability density function}
\acro{PMF}{probability mass function}
\acro{MPS}{matrix product state}
\acro{WEP}{weight enumerator polynomial}
\acro{WE}{weight enumerator}
\acro{QLDPC}{quantum low-density parity check}
\acro{CSS}{Calderbank, Shor and Steane}
\acro{WE}{weight enumerator}
\acro{BD}{bounded distance}

\end{acronym}
\setcounter{page}{1}

\begin{abstract}

In the implementation of quantum information systems, one type of Pauli error, such as phase-flip errors, may occur more frequently than others, like bit-flip errors.
For this reason, quantum error-correcting codes that handle asymmetric errors are critical to mitigating the impact of such impairments. 
To this aim, several asymmetric quantum codes have been proposed. 
These include variants of surface codes like the XZZX and ZZZY surface codes, tailored to preserve quantum information in the presence of error asymmetries.
In this work, we propose two classes of Calderbank, Shor and Steane (CSS) topological codes, referred to as cylindrical and M\"{o}bius codes, particular cases of the fiber bundle family. 
Cylindrical codes maintain a fully planar structure, while M\"{o}bius codes are quasi-planar, with minimal non-local qubit interactions.
We construct these codes employing the algebraic chain complexes formalism, providing theoretical upper bounds for the logical error rate.
Our results demonstrate that cylindrical and M\"{o}bius codes outperform standard surface codes when using the minimum weight perfect matching (MWPM) decoder.
\end{abstract}


\section{Introduction}

The utilization of the distinctive properties of quantum mechanics has significantly broadened the scope of information handling, spanning across sensing, processing, and communication. \cite{Kim:08,Pre:18,OSW:18,NAP:19,WehElkHan:18}. One of the main challenges for quantum information systems concerns mitigating the effects of noise resulting from the inevitable quantum-environment interactions  \cite{Sho:95,Laf:96,Kni:97,FleShoWin:08,RehShi:21}. 
Quantum error correcting codes, with redundant quantum state representations, are therefore crucial for achieving fault-tolerant quantum computation, quantum memories, and quantum communication systems~\cite{Ter:15, MurLiKim:16, Bab:19, Val24:QComm}. 

Stabilizer codes represent a significant class of quantum error-correcting codes~\cite{Got:09}.
Within this category, two overlapping subclasses are of practical interest: topological codes and \ac{CSS} codes \cite{CalSho:96, Ste:96, BraKit:98}.
In fact, despite the increased overhead associated with \ac{CSS} codes compared to their non-\ac{CSS} counterparts, their use simplifies the design of fault-tolerant quantum computing procedures~\cite{Got:96, HorFowDev:12}.
Admitting the use of long-range interactions between qubits, several studies have led to the development of good quantum LDPC codes \cite{BreNikEbe:21a, PanKal22:QLDPC, Lev:QuantumTanner22}.

On the other hand, topological codes offer the advantage of requiring only nearest-neighbor interactions between qubits, a crucial feature for architectures with limited qubit interaction capabilities. 
Some known instances of topological \ac{CSS} codes are the toric and surface codes. 
Among the two, surface codes have attracted more attention due to their planar structure~\cite{FowMarMar:12, FowSteGro:09, Rof:19, Bro:23, ForValChi24:JSAC}.
Finally, the existence of efficient decoding algorithms further ensures its success~\cite{Hig:22, Del:21, ForValChi:24STM}.
To validate this claim, some successful implementations of surface codes have already been achieved~\cite{AchRajAle:22, KriSebLac:22, ZhaYouYe:22, BluDolEve:23, Ach:24}.

In practical scenarios, quantum technologies are more prone to dephasing noise, therefore suffering more phase-flip errors compared to bit-flip ones~\cite{IofMez:07, Ngu:16, Les2020:DephasingNoise, Reg2024:QubitOnlyPhaseFlip}.
Due to this, asymmetric quantum codes have been proposed as a solution for scenarios where strong asymmetries in quantum channel errors arise~\cite{SarKlaRot:09, AtaTucBar:21, ChiVal:20a}. 
In \cite{SarKlaRot:09}, codes are constructed starting from classical codes, such as \acl{BCH} and \acl{LDPC} ones, using the \ac{CSS} construction. 
In \cite{ChiVal:20a}, non-\ac{CSS} asymmetric codes are constructed using the quantum Hamming bound, by syndrome assignment in order to find the shortest possible codes able to guarantee a certain asymmetric error correction capability.
In \cite{AtaTucBar:21}, a variant of surface codes has been proposed, named XZZX surface code. 
At the cost of losing the \ac{CSS} structure of surface codes, the XZZX code exhibits a performance boost in presence of quantum channel asymmetries when compared to the conventional surface code. 
Recently, a variant named ZZZY surface codes, tailored for asymmetric errors has been proposed in \cite{for:24ZZZY}.
On the other hand, to keep a \ac{CSS} structure, hypergraph product code~\cite{TilZem:13} construction and its generalization based on fiber bundle~\cite{HasMatHaa:21} have been proposed.

In this paper we investigate classes of quantum topological and \ac{CSS} codes referred to as cylindrical and M\"{o}bius codes, particular cases of the fiber bundle codes. 
Both codes can be constructed starting from a surface code and attaching two opposite boundaries directly for the cylindrical code or with a half-twist for the M\"{o}bius one.
Cylindrical codes exhibit a two-dimensional topological structure that can be configured over planar lattices, and therefore necessitating only local qubit interactions in two dimensions. 
On the other hand, despite the topological structure of M\"{o}bius codes, in two dimension they are quasi-planar, in the sense that few non-local qubit interactions are required.
In this scenario, a natural implementation choice could be to leverage the inherent mobility offered by reconfigurable atom arrays~\cite{Xu24:reconfigurableAtoms}.
These architectures allow for processor connectivity to be reconfigured during quantum evolution by shuttling atoms around in optical tweezers, with minimal decoherence.
As a result, they are particularly well-suited for realizing a limited number of remote connections~\cite{Xu24:reconfigurableAtoms}.
On the other hand, cylindrical codes can be implemented on superconducting quantum computers as surface codes.

We show that cylindrical and M\"{o}bius codes outperform surface codes with the same number of qubits, both on symmetric and asymmetric channels, using a \acf{MWPM} decoder. 
We analytically quantify this advantage employing the concept of \acf{WE} for the undetectable errors, and we numerically assess the performance of these codes via Monte Carlo simulations. 
The key contributions of the paper can be summarized as follows:

\begin{itemize}
    \item We introduce the cylindrical codes, describing their structure and deriving the code parameters starting from chain complexes formalism. 
    \item By applying twists in the cylindrical codes construction, we derive a variant that we name M\"{o}bius codes. 
    \item We analytically investigate the performance of M\"{o}bius and cylindrical codes over symmetric and asymmetric channels, assuming a \ac{MWPM} decoder,
    for small codes.
    \item We provide analytical bounds exploiting the concept of \ac{WE} for the undetectable errors and the quantum MacWilliams identities, for any code dimension. 
    \item We numerically verify the analysis of the proposed code performance with a \ac{MWPM} decoder over depolarizing and polarizing channels.
\end{itemize}

This paper is organized as follows.
Section~\ref{sec:preliminary} introduces preliminary concepts and models together with some background material. 
In Section~\ref{sec:CylMob} we first describe the structure of cylindrical and M\"{o}bius codes and we derive the code parameters using chain complexes formalism from topology. 
Then we analyse the performance of cylindrical codes through the concept of \ac{WE} for the undetectable errors. 
Numerical results are shown in Section~\ref{sec:NumRes}.

\section{Preliminaries and Background}
\label{sec:preliminary}

\subsection{Stabilizer codes}

A qubit is an element of the two-dimensional Hilbert space $\mathcal{H}^{2}$, with basis $\ket{0}$ and $\ket{1}$ \cite{NieChu:10}. 
The Pauli operators $\M{I}, \M{X}, \M{Z}$, and $\M{Y}$, are defined by  $\M{I}\ket{a}=\ket{a}$, $\M{X}\ket{a}=\ket{a\oplus 1}$, $\M{Z}\ket{a}=(-1)^a\ket{a}$, and $\M{Y}\ket{a}=i(-1)^a\ket{a\oplus 1}$ for $a \in \lbrace0,1\rbrace$. These operators either commute (e.g. $\M{I}\M{X}=\M{X}\M{I}$) or anticommute (e.g. $\M{X}\M{Z}=-\M{Z}\M{X}$) with each other. 
Similarly, when considering Pauli operators on $n$ qubits, along with the same multiplicative factors, one constructs the $\mathcal{G}_n$ Pauli group \cite{Got:09,NieChu:10}.
We indicate with $[[n,k,d]]$ a \ac{QECC} that encodes $k$ information qubits (called logical qubits), into a codeword\footnote{A codeword is any state in the code space.} of $n$ qubits  $\ket{\psi}$ (called data or physical or codeword qubits), able to correct all patterns up to $t = \lfloor(d-1)/2 \rfloor$ errors and, usually, some patterns of more errors. 
Using the stabilizer formalism, we start by choosing $n-k$ independent and commuting operators $\M{G}_i \in \mathcal{G}_n$, called generators, where $\mathcal{G}_n$ is the Pauli group on $n$ qubits \cite{Got:09,NieChu:10}.   
The subgroup of $\mathcal{G}_n$ generated by all combinations of the $\M{G}_i \in \mathcal{G}_n$ is called stabilizer and indicated as $\mathcal{S}$. 
The code $\mathcal{C}$ is the set of quantum states (or codewords) $\ket{\psi}$ stabilized by $\mathcal{S}$, i.e., satisfying 
$\M{S}\ket{\psi}=\ket{\psi} \, \forall \M{S} \in \mathcal{S}$, or, equivalently, $\M{G}_i \ket{\psi}=\ket{\psi},\, i=1, 2, \ldots, n-k$. 
Generators are used to extract the error syndrome by the mean of ancilla qubits.
Finally, the operators $\M{L}$ which commute with the stabilizer group, but $\M{L} \notin \mathcal{S}$, are called logical operators and represent undetectable errors as they turn one codeword into another.

\begin{figure}[t]
    \centering
    \includegraphics[width=\columnwidth]{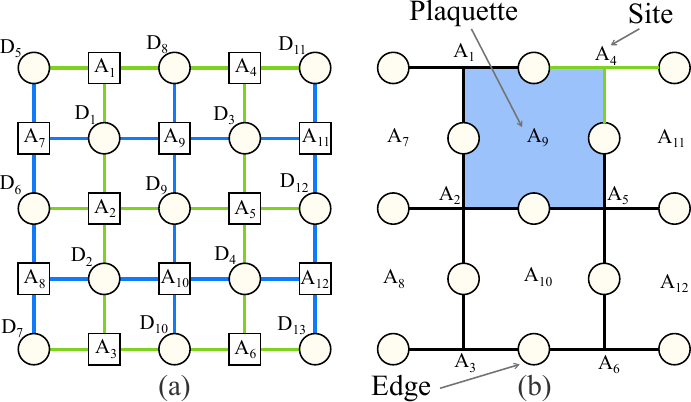}
    \caption{Pictorial representation of a $[[13, 1, 3]]$ surface code. (a) Actual physical representation of the qubits of the code. Data qubits in circles and ancilla qubits on squares. (b) Simplified representation of the same lattice.}
    \label{fig:SurfaceExplained}
\end{figure}

Among stabilizer codes we find the topological codes.
One of the main advantage of these codes is that they require only nearest-neighbor interactions between qubits.
In particular, arranging qubits on a planar grid with boundaries we obtain surface codes~\cite{Rof:19}, which is even more appealing from an architecture perspective.
An example of a $[[13, 1, 3]]$ surface code is depicted in Fig.~\ref{fig:SurfaceExplained}. 
In Fig.~\ref{fig:SurfaceExplained}a the actual physical disposition  where codeword and ancilla qubits are represented in circles and squares, is provided.
Here, the connections between squares and circles represent a particular measurement (green color stands for $\M{X}$ measurement and blue color stands for $\M{Z}$ measurement).
For example, the ancilla $A_1$ measures the codeword according to the $\M{X}_1\M{X}_5\M{X}_8$ generator.
In Fig.~\ref{fig:SurfaceExplained}b we also report a simplified notation having in the edges of a grid the codeword qubits (emphasized with circles), while the ancillary qubits are located both in the vertexes and the faces of the lattice.
In particular, qubits located in the vertexes are called \lq\lq sites\rq\rq, while qubits located in the faces are called \lq\lq plaquettes\rq\rq~\cite{DenKitLan:02, HorFowDev:12}.
Since the surface codes have boundaries, some ancillas (sites or palquettes) will be adjacent to three qubits (edges).
To differentiate them to other ancillas, we define them as boundary sites (e.g., $A_1$, $A_3$, $A_{4}$, and $A_{6}$) and as boundary paquettes (e.g., $A_7$, $A_8$, $A_{11}$, and $A_{12}$).
In surface codes, logical operators can be easily identified: $\M{Z}_L$ ($\M{X}_L$) operator consists of a tensor product of $\M{Z}$'s ($\M{X}$'s) acting on a chain of qubits running from an edge belonging to a boundary plaquette (site) to an edge belonging to a plaquette (site) on the opposite boundary of the lattice.
Finally, the importance of topological codes is enhanced by the availability of efficient decoders \cite{Del:21,Hig:22,ForValChi:24STM}. 
In this paper we consider the \ac{MWPM}. 
This decoder finds the shortest way to connect pairs of ancilla qubits that anticommute with the error, i.e. ancillas that have detected an error.

\subsection{Theoretical Performance Analysis of Quantum Codes}
Considering a quantum stabilizer code, if we measure the codeword according to the generators $\M{G}_i$ with the aid of ancilla qubits, the error collapses on a discrete set of possibilities represented by combinations of Pauli operators $\M{E}\in \mathcal{G}_n$ \cite{Got:09}. 
We refer to this $\M{E}$ as Pauli error. 
The weight of an error $\M{E} \in \mathcal{G}_n$ is the number of single qubits Pauli operators which are not equal to the identity.  
A common channel model is the one characterized by errors occurring independently and with the same statistic on the individual qubits of each codeword. In this model, the error can be $\M{X}$, $\M{Z}$ or $\M{Y}$ with probabilities $\px$, $\pz$, and $\py$, respectively. The probability of a qubit generic error is $p = \px + \pz + \py$.
Two important models are the \emph{depolarizing channel} where $\px = \pz = \py = p / 3$, and the \emph{phase-flip channel} where $p=\pz$, $\px = \py=0$. 
In the case of asymmetric channels with $p_\mathrm{X} = p_\mathrm{Y}$, the asymmetry is described by the bias parameter $A = 2\pz /(p - \pz)$. For instance, in the phase flip channel $A \to \infty$, and the depolarizing channel is obtained for $A=1$. In the following, we will also adopt the notation $[[n, k,d_\mathrm{X}/d_\mathrm{Z}]]$ for asymmetric codes able to correct all patterns up to $t_\mathrm{X} = \lfloor(d_\mathrm{X}-1)/2\rfloor$ Pauli $\M{X}$ errors and $t_\mathrm{Z} = \lfloor(d_\mathrm{Z}-1)/2\rfloor$ Pauli $\M{Z}$ errors. 

We define $f_j(i,\ell)$ as the fraction of errors of weight $j$, with $i$ Pauli $\M{Z}$ and $\ell$ Pauli $\M{X}$ errors, which are not corrected by a complete decoder. Hence, the error probability of a \ac{QECC}, also known as logical error rate, can be written as \cite{ForValChi:23}
\begin{align}
\label{eq:rho_L}
    p_L = \sum_{j = 0}^{n} \binom{n}{j}(1-p)^{n-j} p^j (1-\beta_j)
\end{align}
with
\begin{align}
    1-\beta_j = \frac{1}{p^j}\sum_{i = 0}^{j}\binom{j}{i} \, \pz^i \, \sum_{\ell = 0}^{j-i} \binom{j-i}{\ell}\, \px^\ell \, \py^{j-i-\ell} f_j(i,\ell)\,
\end{align}
For example, for the depolarizing channel, $\beta_j$ can be interpreted as the fraction of errors of weight $j$ that a complete decoder is able to correct. 
In general, $\beta_j$ depends on the code structure, on the decoder, and on the channel asymmetry parameter $A$. 
Starting from \eqref{eq:rho_L} we can approximate the logical error rate
for $p \ll 1$ as
\begin{align}
\label{eq:error_probWithBetaApprox}
p_\mathrm{L} 
&\approx \left(1-\beta_{t+1}\right) \binom{n}{t+1}p^{t+1} \,.
\end{align}
Hence, given $\beta_{t+1}$, we can evaluate the asymptotic logical error rate of a quantum code.
In a similar way, we can find the asymptotic performance of an asymmetric quantum $[[n, k, d_X/d_Z]]$ code by simply set $t = \min(t_\mathrm{X},t_\mathrm{Z})$. 
Note that, $\max(t_\mathrm{X},t_\mathrm{Z})$ affects the term $f_j(i,\ell)$ inside $\beta_j$.

\subsection{Undetectable errors weight enumerator from Quantum MacWilliams identities}

The values of $\beta_{j}$ for the specific decoder and quantum channel can be evaluated by exhaustive search or ad-hoc reasoning on the logical operators~\cite{ForValChi:23}. 
Despite of the adopted method, when these values are tabulated anyone can plot the actual performance without implementing the decoder.
To aid the derivation of $\beta_{j}$, alongside a $[[n,k,d]]$ quantum code, we will provide the undetectable errors \ac{WE} polynomial as 
\begin{align}
\label{eq:deflogicals_eq}
L(z) &= \sum_{w = 0}^{n} L_w z^w
\end{align} 
where $L_w$ is the number of undetectable errors (logical operators) of weight $w$. 
This polynomial can be computed starting from the code's generators by using the quantum MacWilliams identities \cite{ShoPetLaf:97}. 
Specifically, for an $[[n,k,d]]$ stabilizer code we have that the undetectable errors \ac{WE} is defined as~\cite{ForValChi:23}
\begin{align}
\label{logicals_eqMC}
L(z) = \frac{1}{2^k} B(z) - \frac{1}{4^k} A(z) 
\end{align} 
%
where $\frac{1}{4^k} A(z)$ and $\frac{1}{2^k} B(z)$ are equal to the stabilizer \ac{WE} and the normalizer \ac{WE} of the code, respectively \cite{CaoChuLac:22}. 
Considering the connection between stabilizer codes and codes over $\mathbb{F}_4$, the evaluation of $A(z)$ can be seen as the computation of the weight distribution of classical codes. 
To find such distribution, it is possible to exploit software tools related to coding theory, such as MAGMA \cite{BosWieCan:97}.
Given $A(z)$, we can obtain $B(z)$ using the quantum MacWilliams identities \cite{ ShoPetLaf:97, Rai:98} 
\begin{align}
B_w &= \frac{1}{2^n} \sum_{\ell = 0}^{n} \; \sum_{s = 0}^{w} \;\; \binom{\ell}{s} \binom{n-\ell}{w-s} (-1)^s 3^{w-s}A_\ell \,.
\end{align}
Then, from \eqref{logicals_eqMC} we derive the undetectable errors \ac{WE}, $L(z)$.

One approach that can be taken to derive the performance of a code is the logical operator analysis.
In fact, using $L(z)$ together with some understanding of the structure of the code's logical operators, ad-hoc counting could be performed on small codes to obtain the exact value of $\beta_{j}$.
In this procedure, given the total number of logical operators $L_w$ of weight $w$, the first step is to identify all such operators in the code under examination. 
These operators are then grouped into sets based on their similar structure, which typically results in similar performance behavior, often due to inherent symmetries.
Finally, considering that $j$ errors are introduced by the channel, we search for the error patterns that could trigger one of the logical operator of weight $w$ when the decoder introduces $w-j$ errors.
Note that, during this evaluation, it is important to remember that logical operators may share common operators on the same qubits. 
Therefore, extra care must be taken to avoid double counting.
This approach will be used, together with several explanatory examples, in Section~\ref{subsec:CylPerfAnalysis} and Section~\ref{subsec:MobPerfAnalysis} to evaluate the exact value of $\beta_2$ for both the cylindrical and M\"{o}bius codes with $d=3$.

For large codes, this approach becomes infeasible. 
Therefore, bounds that rely solely on the knowledge of $L(z)$ and the code's structure can be used to estimate its performance.
This will be addressed for both the cylindrical and M\"{o}bius codes in Section~\ref{subsec:UpperBounds}, to estimate the performance for any code dimension.

\subsection{Topological Interpretation of Linear Codes} \label{sec:QLDPC_CC}

In the following, we introduce a few concepts from topology that we later use to build quantum cylindrical codes and derive their parameters. 

A \textit{chain complex} $C$ is a collection of vector spaces $\{C_i\}$ for $i = 0,1,...,m$,   and linear maps \cite{BreGle:13, HatAll:05}
\begin{equation}
\label{eq:cchains}
\partial_i : C_i \xrightarrow{} C_{i-1}
\end{equation}
with the requirement that $\partial_i \partial_{i+1} = 0$, or, equivalently, that $ \image\, \partial_{i+1} \subseteq \ker \partial_i$. Note that the trivial operators $\partial_0 : C_0 \xrightarrow{} \{0\}$ and $\partial_{m+1} : \{0\} \xrightarrow{} C_m$ are implicit. 
In particular, each vector space has a canonical basis, which elements are called \textit{i-cells}, while vectors of $C_i$ are named \textit{i-chains}.
It is possible to express the \textit{i-th homology} as the quotient space
%
%
%
\begin{equation}
\label{eq:hom}
H_i\left(C\right) = \frac{\ker \partial_i}{\image\, \partial_{i+1} }
\end{equation}
where elements of $\ker \partial_i$ and $\image\, \partial_{i+1} $ are called \textit{i-cycles} and  \textit{i-boundaries}, respectively. Each chain complex has a dual chain, or \textit{cochain}. This is defined as the sequence of dual vector spaces $C_{i}^{\ast}$ (the space of linear functionals of $C_{i}$) and dual maps
\begin{equation}
\label{eq:cochains}
 \partial^\ast_{i+1}: C_{i+1}^{\ast} \leftarrow{} C_{i}^{\ast}.
\end{equation}
The \textit{i-th cohomology} is
\begin{equation}
\label{eq:cohom}
H^i\left(C\right) = \frac{\ker \partial_{i+1}^\ast}{\image\, \partial_{i}^\ast }
\end{equation}
where $\dim H_i(C) = \dim H^i(C)$. 
Moreover, we define
\begin{align}
    \xi_i &\triangleq \min \left\{ 
     f_i(\V{c}) \,|\, \V{c} \in \ker \partial_i \setminus \image\, \partial_{i+1} 
    \right\}\\
    \zeta_i &\triangleq \min \left\{ 
     f_i(\V{c}) \,|\, \V{c} \in \ker \partial_{i+1}^\ast \setminus \image\, \partial_{i}^\ast  
    \right\} 
\end{align}
where $f_i(\V{c}): C_i \to \mathbb{N}$ is a function that counts how many formally summed terms are in $\V{c} \in C_i$. For $C_i \subseteq \mathbb{F}_2^m$, $f_i(\cdot)$ coincides with the Hamming weight function $w_\mathrm{H}(\cdot)$ on a binary $m$-tuple.

\begin{example}[\emph{Topological interpretation of classical linear codes}]
Any classical $[n,k,d]$ binary linear code can be seen as a 2-term chain complex 
\begin{equation}
\label{eq:ccode}
C = C_1 \xrightarrow{\partial_1} C_0
\end{equation}
where $\partial_1 = H \in \mathbb{F}_2^{r\times n}$, with $r \ge n-k$ and $\rank(H) = n-k$, is the linear map given by the parity check matrix, $C_1 = \mathbb{F}_2^n$, and $C_0 = \mathbb{F}^{r}_2$. 
We can compute $H_1(C) = \ker \partial_1$ and $\xi_1 = \min \left\{ w_\mathrm{H}(\V{c}) \,|\, \V{c} \in \ker \partial_1 \right\}$.
By definition, we have that $\ker \partial_1 \subseteq \mathbb{F}_2^n$ is the space of the codewords which implies that $k = \dim( \ker \partial_1 ) = \dim(H_1(C))$, $\image\, \partial_1$ is the space of the parity checks (or error syndromes), and the code minimum distance is $d = \xi_1$.
Finally, we point out that, when $H$ is full rank, we have $r = n-k$, $\image\, \partial_1 = \mathbb{F}^{n-k}_2$, and $H_0(C) = \mathbb{F}^{n-k}_2 / \mathbb{F}^{n-k}_2 \cong 0$.
On the other hand, when $H$ is not full rank, $r > n-k$, $\image\, \partial_1 \subset \mathbb{F}^{r}_2$ and in particular $ \dim (\image\, \partial_1) = n -k $, and $H_0(C) = \mathbb{F}^{r}_2 / \image\, \partial_1 \cong {\mathbb{F}_2}^{r - n + k} $.
\end{example}

\begin{lemma}[\emph{\ac{CSS} binary construction~\cite{SarKlaRot:09}}]
Consider two linear codes $\mathcal{C}_x$ and $\mathcal{C}_z$ with parameters $[n,k_x]$ and $[n,k_z]$, respectively. If $\mathcal{C}^\perp_x \subseteq \mathcal{C}_z$ there exists an asymmetric $[[n, k_x+k_z-n, d_\mathrm{X}/d_\mathrm{Z}]]$ quantum code where $d_\mathrm{X} = \min \left\{ 
     w_H(\V{c}) \,|\, \V{c} \in \mathcal{C}_x \setminus \mathcal{C}^\perp_z
\right\}$ and $d_\mathrm{Z} = \min \left\{ 
     w_H(\V{c}) \,|\, \V{c} \in \mathcal{C}_z \setminus \mathcal{C}^\perp_x
\right\}$. 
\end{lemma} 

Similarly to classical codes, the chain complex formalism can be used to describe \ac{CSS} quantum codes \cite{BraSerHas:14, HasMatHaa:21, StrArmBer:22, TilZem:13, ZenWeiPry:19}. 
The following example illustrates this interpretation in detail.

\begin{example} [\emph{Topological interpretation of quantum CSS codes}]
In general, a \ac{CSS} code corresponds to a three terms chain complex 
\begin{equation}
\label{eq:csscode}
C = C_{2} \xrightarrow{\partial_{2}} C_{1} \xrightarrow{\partial_{1}} C_{0}
\end{equation}
where $\partial_{2} = H^{\top}_Z \in \mathbb{F}^{n \times r_z}_2$, $\partial_{1} = H_X \in \mathbb{F}^{r_x \times n}_2$, $r_z \geq \rank(H_Z) = n - k_z$, $r_x \geq \rank(H_X) = n - k_x$, $C_2 = \mathbb{F}^{r_{z}}_2$, $C_1 = \mathbb{F}^n_2$, and $C_0 =\mathbb{F}^{r_{x}}_2$.
In this case, $H_1(C) = \ker \partial_1 / \image\, \partial_2$, $\xi_1 =\min \left\{ w_\mathrm{H}(\V{c}) \,|\, \V{c} \in \ker \partial_1 \setminus \image\, \partial_2 \right\}$, $H^1(C) = \ker \partial^\top_2 / \image\, \partial^\top_1$ and $\zeta_1 =\min \left\{ w_\mathrm{H}(\V{c}) \,|\, \V{c} \in \ker \partial^\top_2 \setminus \image\, \partial^\top_1 \right\}$. 
As regards the code parameters, we have by definition that $n = \dim( C_1 )$, $\ker \partial_1 \subseteq \mathbb{F}_2^n$ is the space of the codewords on which act only $\M{X}$ parity checks $(\mathcal C_x)$, $\ker \partial^\top_2 \subseteq \mathbb{F}_2^n$ is the space of the codewords checked by $\M{Z}$ stabilizers $(\mathcal C_z)$, and $k = \dim( H_1 ) = \dim( H^1 )$. In addition, $\image\, \partial_1$ is the space of the $\M{X}$ generators and $\image\, \partial^\top_2$ is the space of the $\M{Z}$ stabilizers. 
The code minimum distance are $d_X = \zeta_1$ and $d_Z = \xi_1$, corresponding to the minimum weight of a nontrivial representative of $H^1$ and $H_1$. 
Note that, if we ensure $\partial_i \partial_{i+1} = 0$ by construction, then $\mathcal C_x^\perp \subseteq \mathcal C_z$. Hence, the chain stands for a valid CSS code. On the other hand, if we choose codes such that $\mathcal C_x^\perp \subseteq \mathcal C_z$ we obtain a valid chain complex. This proves that this structure can be used to represent any CSS code. \\
\end{example}

To construct a chain complex representing a \ac{CSS} code it is common to use double complexes of a total complex. For any two chain complexes $C$ and $D$, of length $M$ and $N$, it is possible to define the \textit{double complex} $C~\boxtimes~D$ as \cite{BreNikEbe:21a, BreNikEbe:21b}
\begin{equation}
\label{eq:doublec}
\left(C \boxtimes D\right)_{p,q} = C_p \otimes D_q
\end{equation}
where $p = 0,1,...,M$ and $q = 0,1,...,N$. In this construction we have two types of boundary maps: $\partial^{v}_i = \partial^{C}_i \otimes I^D$ and  $\partial^{h}_i = I^{C} \otimes \partial^{D}_i$, such that $\partial^v_i \partial^v_{i+1} = 0$, $\partial^h_i \partial^h_{i+1} = 0$ and $\partial^v_i \partial^h_j = \partial^h_j \partial^v_i$.

We can collect vector spaces of equal dimensions by summing along the diagonals, in order to obtain the \textit{total complex} 
\begin{equation}
\label{eq:totalc}
 E_n = Tot\left(C \boxtimes D\right)_n = \bigoplus_{p+q=n} C_p \otimes D_q 
\end{equation}
where $n = 0,1,...,N + M$.
The resulting boundary maps are $\partial^E = \partial^v \oplus \partial^h$. Finally, we can obtain a new chain complex, called \textit{tensor product complex}, starting from $C$ and $D$ as
\begin{equation}
\label{eq:tensprodc}
E = C \otimes D = Tot\left(C \boxtimes D\right).
\end{equation}
Moreover, the Künneth formula gives a method to compute the homology of a tensor product complex from the homology of the original chains
\begin{equation}
\label{eq:kun}
H_n \left(C \otimes D\right) \cong \bigoplus_{p+q=n} H_p(C) \otimes H_q(D).
\end{equation}

In the following, we will use these concepts to describe the cylindrical and M\"{o}bius codes.

\section{Cylindrical and M\"{o}bius Codes}
\label{sec:CylMob}

\begin{figure*}[t]
	\centering
	\includegraphics[width = \textwidth]{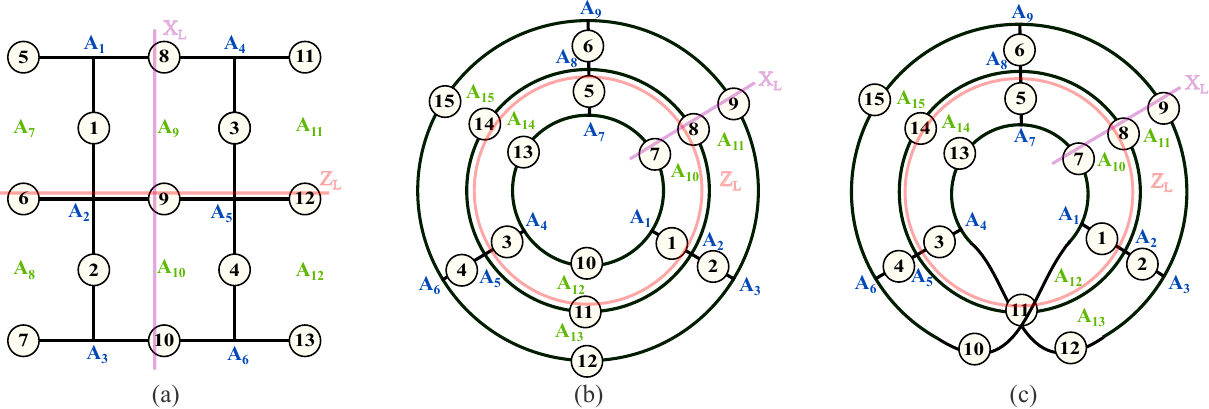}
	\caption{ (a) $[[ 13,1,3 ]]$ surface code. Data qubits are depicted as circles, blue ancillas represent $\M{Z}$ stabilizers while red ancillas stand for $\M{X}$ stabilizers. Examples of $\mathbb{\M{X}_L}$ and $\mathbb{\M{Z}_L}$ logical operators are drawn on the lattice. (b) $[[ 15,1,3 ]]$ cylindrical code. (c) $[[ 15,1,3 ]]$ M\"{o}bius code.} 
	\label{Fig:cylinder_surface}
\end{figure*}

\subsection{Cylindrical Codes: Design using Topology} 

As shown in section~\ref{sec:QLDPC_CC}, it is possible to describe a quantum \ac{CSS} code using a $3$-term chain.
Here we report how to construct \ac{CSS} codes starting from $2$-term chain complexes representing classical linear binary codes.
In particular, we detail the construction of the more general family of hypergraph product codes~\cite{TilZem:13}, within which the cylindrical code is included.

Let us denote with $C$ and $F$ their respective chain complexes
\begin{align}
\label{eq:C}
    C &= C_1 \xrightarrow{H_C} C_0\\
    F &= F_1 \xrightarrow{H_F} F_0
\end{align}
where $H_C \in \mathbb{F}_2^{r_c \times n_c}$ and $H_F \in \mathbb{F}_2^{r_f \times n_f}$ are the code parity check matrices.
At this point, we take the dual of $F$ obtaining the cochain complex
\begin{align}
    F^* &= F^*_0 \xrightarrow{H^*_F}  F^*_1 \cong F_0 \xrightarrow{H^\top_F} F_1\,.
\end{align}
Considering the isomorphism between chain and cochain, we rewrite $F^*$ as a chain complex $D$, resulting in
\begin{align}
\label{eq:D}
    D &= D_1 \xrightarrow{H^\top_F} D_0\,.
\end{align}
Having $C$ and $D$ representing our two initial codes, we construct the double complex according to \eqref{eq:doublec} as
\[
\begin{tikzcd}[row sep=1.2cm,column sep=0.8cm]
C_1 \otimes D_1  \arrow[r,"I_{n_c} \otimes H^\top_F"] \arrow[d, "H_C \otimes I_{r_f} "] & C_1 \otimes D_0   \arrow[d, "H_C \otimes I_{n_f} "] \\
C_0 \otimes D_1 \arrow[r,"I_{r_c}\otimes H^\top_F"] & C_0 \otimes D_0 .
\end{tikzcd}
\]
The tensor product complex $E$ assumes the form
\begin{equation}
\label{eq:CSSChainConstructed}
 E = \underbracket{C_1 \otimes D_1}_{E_2} \xrightarrow{\partial^E_2} \underbracket{C_0 \otimes D_1 \oplus C_1 \otimes D_0}_{E_1} \xrightarrow{\partial^E_1} \underbracket{C_0 \otimes D_0}_{E_0}
\end{equation}
where
\begin{align}
\label{eq:tensprodc3}
\partial^E_2 &= H^{\top}_Z = \left( \begin{array}{c} H_C \otimes  I_{r_f} \\ I_{n_c} \otimes H^{\top}_F \end{array} \right) \\
\label{eq:tensprodc4}
\partial^E_1 &= H_X = \left( I_{r_c} \otimes H^{\top}_F  | \; H_C \otimes  I_{n_f} \right).
\end{align}
Since a tensor product complex is a complex chain, we have a valid CSS by construction (i.e., $\partial_1^{E} \partial_2^{E}$). We can also verify that
\begin{align}
\partial^E_1 \partial^E_2 &= H_X H^{\top}_Z = \left( H_C \otimes H^\top_F \right) + \left( H_C \otimes H^\top_F \right) = \M{0}.
\end{align}
Finally, adopting the following convention for the final \ac{CSS} generator matrix 
\begin{align}
    H = \begin{pmatrix}
        H_X & 0 \\
        0 & H_Z \\
    \end{pmatrix}
\end{align}
we obtain that ancillas with indexes from $1$ to $r_c n_f$ are site generators ($\M{X}$ generators), while ancillas with indexes from $r_c n_f + 1$ to $r_c n_f + n_c r_f$ are plaquette generators ($\M{Z}$ generators). 
This notation will be consistently used throughout the paper when providing examples.

We use the developed topological framework to describe and evaluate the characteristics of the cylindrical code. 
To glue the boundaries of a surface code in only one direction and obtain a structure homeomorphic to a cylinder (i.e., an annulus in 2D), we choose two repetition codes, one with full rank parity check matrix and one with a square parity check matrix.
Specifically, the complex chain $C$ represents a repetition code $[L,1,L]$ with $L$ parity checks. Its homologies are $H_1(C) \cong \mathbb{F}_2$ and $H_0(C) \cong \mathbb{F}_2$, since one of the checks is linearly dependent. Furthermore, the complex chain $F$ is also a repetition code $[L,1,L]$, but with $L-1$ parity checks. In this case, the code has full rank parity check matrix and the homologies are  $H_1(F) \cong \mathbb{F}_2$ and $H_0(F) \cong 0$, since all the checks are linearly independent. Using~\eqref{eq:kun} it is straightforward to find the number of logical qubits encoded by the tensor product code: $k = \dim H_1(E) = \dim(H_0(C) \otimes H_1(D) \oplus H_1(C) \otimes H_0(D)) = 1 \cdot 1 + 1 \cdot 0 = 1$. Moreover, $n = \dim E_1 = \dim(C_0 \otimes D_1 \oplus C_1 \otimes D_0) = L \cdot (L-1) + L \cdot L $, while the distance of the code is still $L$. The resulting \ac{CSS} code has parameters $[[ L^2 + L \cdot (L-1), 1, L]]$.

For example, using $L=3$ and 
\begin{align}
    H_\mathrm{C} = \begin{pmatrix}
        1 & 1 & 0\\
        0 & 1 & 1\\
        1 & 0 & 1\\
    \end{pmatrix} \quad H_\mathrm{F} = \begin{pmatrix}
        1 & 1 & 0\\
        0 & 1 & 1\\
    \end{pmatrix}\notag
\end{align}
we obtain the generators of the $[[15,1,3]]$ cylindrical code through \eqref{eq:tensprodc3} and \eqref{eq:tensprodc4} as 

\begin{equation*}
\small
\arraycolsep=1.6pt
\begin{array}{lllllll} 
\M{G}_1 &= \mathbf{X}_1 \mathbf{X}_7 \mathbf{X}_{10} &
\M{G}_2 &= \mathbf{X}_1 \mathbf{X}_2 \mathbf{X}_8 \mathbf{X}_{11} &
\M{G}_3 &= \mathbf{X}_2 \mathbf{X}_9 \mathbf{X}_{12} &\\
\M{G}_4 &= \mathbf{X}_3 \mathbf{X}_{10} \mathbf{Z}_{13} &
\M{G}_5 &= \mathbf{X}_3 \mathbf{X}_4 \mathbf{X}_{11} \mathbf{X}_{14}  &
\M{G}_6 &= \mathbf{X}_4 \mathbf{X}_{12} \mathbf{X}_{15} \\
\M{G}_7 &= \mathbf{X}_5 \mathbf{X}_7 \mathbf{X}_{13}  &
\M{G}_8 &= \mathbf{X}_5 \mathbf{X}_6 \mathbf{X}_8 \mathbf{X}_{14} &
\M{G}_9 &= \mathbf{X}_6 \mathbf{X}_9 \mathbf{X}_{15} &\\
\M{G}_{10} &= \mathbf{Z}_1 \mathbf{Z}_{5} \mathbf{Z}_{7} \mathbf{Z}_{8}&
\M{G}_{11} &= \mathbf{Z}_2 \mathbf{Z}_{6} \mathbf{Z}_{8} \mathbf{Z}_{9} &
\M{G}_{12} &= \mathbf{Z}_{1} \mathbf{Z}_3 \mathbf{Z}_{10} \mathbf{Z}_{11} \\
\M{G}_{13} &= \mathbf{Z}_{2} \mathbf{Z}_{4} \mathbf{Z}_{11} \mathbf{Z}_{12} &
\M{G}_{14} &= \mathbf{Z}_{3} \mathbf{Z}_{5} \mathbf{Z}_{13} \mathbf{Z}_{14}&
\M{G}_{15} &= \mathbf{Z}_{4} \mathbf{Z}_{6} \mathbf{Z}_{14} \mathbf{Z}_{15} \,. &
\end{array} 
\normalsize
\end{equation*}

%
The structure of cylindrical codes can be visualized by gluing together two plaquette (or site) boundaries of a surface code (see Fig.~\ref{Fig:cylinder_surface}a), including $d - 1$ additional qubits, obtaining an annulus (see Fig.~\ref{Fig:cylinder_surface}b). 
Note that, the qubit indexing of the surface code in Fig.~\ref{Fig:cylinder_surface}a is obtained with 
\begin{align}
    H_\mathrm{C} = H_\mathrm{F} = \begin{pmatrix}
        1 & 1 & 0\\
        0 & 1 & 1\\
    \end{pmatrix}.\notag
\end{align}
It is important to note that cylindrical codes are still planar and their generators require only local connectivity, as in surface codes.

\subsection{M\"{o}bius Codes: Design using Topology} 

A possible method to close in a loop a surface code, although admitting $2 (d-1)$ non-local measurements, is to attach the boundary together with a twist. Due to the particular construction we name these codes as M\"{o}bius codes. An example is depicted in Fig.~\ref{Fig:cylinder_surface}c. Before proceeding with the construction, let us recall some concepts from the \emph{fiber boundle} literature~\cite{HasMatHaa:21, BreNikEbe:21a}.
In fact, the M\"{o}bius codes are particular cases of the family of fiber bundle codes~\cite{HasMatHaa:21}.

Let us consider the complex chains $C$ and $D$ in \eqref{eq:C} and \eqref{eq:D}. 
Denote by $\text{Aut}(D)$ the finite group of linear automorphisms of the complex $D$, i.e. linear automorphisms of $D_1$ and $D_0$ that commute with the differential $H_F^\top$. 
Moreover, denote basis vectors of $C_i$ by $c^i$ and write $c^0 \in H_C \, c^1$ if $c^0$ appears with a nonzero coefficient in $H_C c^1$. 
It is possible to twist the vertical differentials in the double complex $C \boxtimes D$ by an automorphism $\phi(c^1,c^0) \in \text{Aut}(D)$ to every pair $(c^1,c^0)$ such that $c^0 \in H_C \, c^1$. 
The resulting fiber boundle double complex $C~\boxtimes_{\phi}~D$ is 
\[
    \begin{tikzcd}[row sep=1.2cm,column sep=0.8cm]
C_1 \otimes D_1  \arrow[r,"I_{n_c} \otimes H^\top_F"]  \arrow[d,"\partial_{\phi_1}"] & C_1 \otimes D_0   \arrow[d,"\partial_{\phi_0}"] \\
C_0 \otimes D_1 \arrow[r,"I_{r_c} \otimes H^\top_F"]  & C_0 \otimes D_0 .
\end{tikzcd}
\]
where $\partial_{\phi_i}(c^1 \otimes d^i) = \sum_{c^0 \in H_C \, c^1} c^0 \otimes  \phi_i(c^1,c^0)(d^i)$ and $\phi_i(c^1,c^0) \in \text{Aut}(D_i)$. 

Aiming to construct a topological code having a M\"{o}bius structure, we define as: 
$i)$ $P_\mathrm{plaq}$, a $(L-1) \times (L-1)$ matrix with all elements to zero, except for the elements in the secondary diagonal (also called the anti-diagonal or counter-diagonal), which are ones;
$ii)$ $P_\mathrm{site}$, a $L \times L$ matrix with all elements to zero, except for the elements in the secondary diagonal, which are ones; 
$iii)$ $S_x$, a $L \times L$ matrix with all elements to zero, except for the element in position $(x, x)$, which is a one\footnote{Here, we assume that both the rows and columns are indexed starting from one.}.
The matrices $P_\mathrm{plaq}$ and $P_\mathrm{site}$ are permutation matrices we use to produce the twist of the plaquettes and sites, respectively. 
The matrix $S_x$ is used to select a single location to perform the cut of the cylindrical structure, enabling the twist operation.
Then, a M\"{o}bius code is derived from the fiber boundle complex  $C~\boxtimes_{\phi}~D$ imposing 
\begin{equation}
\begin{aligned}
    \partial_{\phi_0} &= (H_C - S_{(L+1)/2}) \otimes I_{n_f} + S_{(L+1)/2} \otimes P_\mathrm{site} \\
    \partial_{\phi_1} &= (H_C - S_{(L+1)/2}) \otimes I_{r_f} + S_{(L+1)/2} \otimes P_\mathrm{plaq} \,.\\
\end{aligned}
\end{equation}
Finally, the construction proceed as usual, obtaining 
\begin{align}
\label{eq:tensprodc3Fiber}
\partial^E_2 &= H^{\top}_Z = \left( \begin{array}{c} \partial_{\phi_1} \\ I_{n_c} \otimes H^{\top}_F \end{array} \right) \\
\label{eq:tensprodc4Fiber}
\partial^E_1 &= H_X = \left( I_{r_c} \otimes H^{\top}_F  | \; \partial_{\phi_0} \right).
\end{align}

Note that, the defined $P_\mathrm{plaq}$, $P_\mathrm{site}$, and $S_x$ matrices correctly construct a M\"{o}bius code if $H_C$ and $H_F$ are defined accordingly.
In particular, take the binary vector $\V{v}$ of length $L$ in which the first two entries are ones and the other entries are zeros. Therefore, $H_C$ and $H_F$ should be constructed having as a first row the vector $\V{v}$, and the other rows are obtained performing a cyclic shift to the right of the above rows.
Then, $H_C$ and $H_F$ should be constructed with the first row being the vector $\V{v}$, the second row obtained by performing a cyclic shift to the right on the above row, and so on.
As an example, using $L=3$, we can construct the $[[15,1,3]]$ M\"{o}bius code adopting
\begin{align}
    H_\mathrm{C} = \begin{pmatrix}
        1 & 1 & 0\\
        0 & 1 & 1\\
        1 & 0 & 1\\
    \end{pmatrix} \quad H_\mathrm{F} = \begin{pmatrix}
        1 & 1 & 0\\
        0 & 1 & 1\\
    \end{pmatrix}\notag
\end{align}
similarly to the cylindrical code.
Then, considering that

\begin{align}
    P_\mathrm{plaq} = \begin{pmatrix}
        0 & 1 \\
        1 & 0 \\
    \end{pmatrix}\quad 
    P_\mathrm{site} = \begin{pmatrix}
        0 & 0 & 1 \\
        0 & 1 & 0 \\
        1 & 0 & 0 \\
    \end{pmatrix} \quad
    S_1 = \begin{pmatrix}
        0 & 0 & 0\\
        0 & 1 & 0\\
        0 & 0 & 0\\
    \end{pmatrix} \notag
\end{align}
we obtain from \eqref{eq:tensprodc3Fiber} and \eqref{eq:tensprodc4Fiber} the same generators of the cylindrical code, except for 

\begin{equation*}
\small
\arraycolsep=1.6pt
\begin{array}{lllllll} 
\M{G}_4 &= \mathbf{X}_3 \mathbf{X}_{12} \mathbf{Z}_{13} &
\M{G}_6 &= \mathbf{X}_4 \mathbf{X}_{10} \mathbf{X}_{15} \\
\M{G}_{12} &= \mathbf{Z}_{1} \mathbf{Z}_{4} \mathbf{Z}_{10} \mathbf{Z}_{11} &
\M{G}_{13} &= \mathbf{Z}_{2} \mathbf{Z}_{3} \mathbf{Z}_{11} \mathbf{Z}_{12} \\
\end{array} 
\normalsize
\end{equation*}

\subsection{Cylindrical Codes: Performance Analysis}
\label{subsec:CylPerfAnalysis}

For the $[[15,1,3]]$ cylindrical code, we compute the undetectable errors \ac{WE} $L(z)$ 
\begin{align}
\label{eq:LzCyl}
    \nonumber L(z) &= 6z^3 + 18z^4 + 66z^5 + 228z^6 + 678z^7 + 1836z^8 \\ \nonumber &+ 4236z^9 + 7920z^{10} + 11274z^{11} + 11442z^{12} \\ &+ 7746z^{13} + 3132z^{14} + 570z^{15}. 
\end{align}
Moreover, for comparison purposes, we also report the expression of $L(z)$ for the $[[13,1,3]]$ surface code
\begin{align}
    \nonumber &L(z) = 6 z^3 + 24 z^4 + 75 z^5 + 240 z^6 + 648 z^7 + 1440 z^8 \\ &+ 2538z^9 + 3216z^{10} + 2634z^{11} + 1224z^{12} + 243z^{13}.  
\end{align}
\begin{figure}[t]
	\centering
	\includegraphics[width = 0.9\columnwidth ]{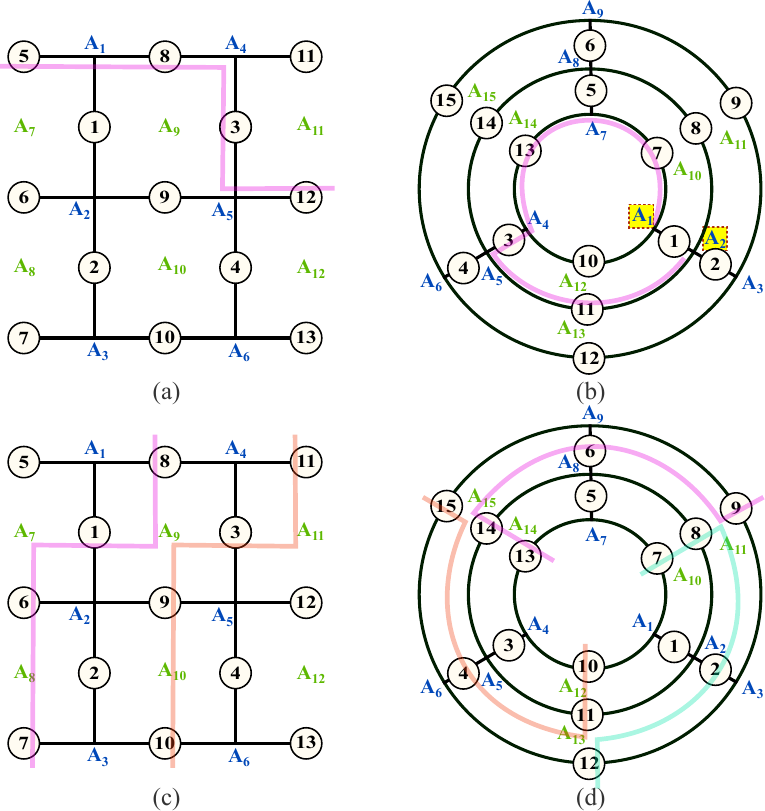}
	\caption{ Some comparison between logical operators on the surface and the cylindrical codes. The pattern highlighted in (b) is not a logical operator. In (a) and (b) error patterns are composed by $\M{Z}$ operators, while in (c) and (d) are composed by $\M{X}$ operators.} 
	\label{fig:SurfVsCylExamples}
\end{figure}

Upon closer examination of these expressions, it is evident that the cylindrical code features six fewer logical operators with weight $w=4$ than the surface code. 
Indeed, for the $[[13,1,3]]$ surface code, we find eight $\M{Z}\M{Z}\M{Z}\M{Z}$ 
logical operators of weight $w = 4$ that cross the lattice from boundary to boundary~\cite{ForValChi:23}. 
The main advantage of the $[[15,1,3]]$  cylindrical code is that these are no longer logical operators, since boundaries are periodic. 
As an example, using notation from Fig.~\ref{Fig:cylinder_surface}a, the $\mathbb{\M{Z}}_L$ logical operator $\M{Z}_3\M{Z}_{5}\M{Z}_{8}\M{Z}_{12}$ has no counterpart in the cylindrical code. 
This is due to the fact that an error of the kind $\M{Z}_3\M{Z}_{7}\M{Z}_{11}\M{Z}_{13}$ in the cylindrical code turns on both ancillas $A_1$ and $A_2$ (see Fig.~\ref{fig:SurfVsCylExamples}a and Fig.~\ref{fig:SurfVsCylExamples}b). 
Moreover, in the surface code we have two logical operators of weight $2t + 2$ of the kind $\M{Y}\M{Y}\M{Z}\M{X}$, such as $\M{Z}_1\M{Y}_5\M{X}_7\M{Z}_6$, that are no longer present in the cylindrical structure.
However, in the $[[15,1,3]]$ cylindrical code we can find four $\M{X}_L$ logical operators of weight $w = 4$ which are not present in the $[[13,1,3]]$ surface code and involve the two additional qubits.
For instance, there are two logical operators, i.e., $\M{X}_{1}\M{X}_{6}\M{X}_{7}\M{X}_{8}$ and $\M{X}_{3}\M{X}_{9}\M{X}_{10}\M{X}_{11}$, crossing the surface code from the bottom boundary to the top one, each making a single rightward turn at the upper part of the lattice (see Fig.~\ref{fig:SurfVsCylExamples}c). 
As anticipated, looking at the cylindrical code, there are three such patterns, i.e., $\M{X}_{6}\M{X}_{9}\M{X}_{13}\M{X}_{14}$, $\M{X}_{2}\M{X}_{7}\M{X}_{8}\M{X}_{12}$ and $\M{X}_{4}\M{X}_{10}\M{X}_{11}\M{X}_{15}$ (see Fig.~\ref{fig:SurfVsCylExamples}d).
Since a path crossing four qubits from bottom to top can have the turn placed either at the upper or lower part, and it can be directed to the right or left, there is a difference of four logical operators of this type.

We now discuss the ad-hoc logical operator analysis, a procedure applicable to small codes to obtain exact asymptotic performance. 
For larger codes, we defer to Section~\ref{subsec:UpperBounds}, where upper bounds are derived using the scalable aspects of the ad-hoc reasoning presented here.
To compute the value of $\beta_2$ of the $[[15,1,3]]$ cylindrical code, we focus on the logical operators of weight $w=3$ and $w=4$, in order to find the fraction of faulty errors of weight $j=2$. 
We can stop at $w=4$ because when $j=2$ errors are introduced by the quantum channel, since at least a codeword is at distance $j$ (the correct one), a decoder choosing the minimum weight pattern could fail introducing at most $j$ errors.

From~\eqref{eq:LzCyl}, we have six logical operators with $w=3$.
Specifically, we have three $\M{X}\M{X}\M{X}$ and three $\M{Z}\M{Z}\M{Z}$ logical operators (two examples can be seen in Fig.~\ref{Fig:cylinder_surface}b), all of which exhibit channel errors triggering a failure when composed of two $\M{X}$ or two $\M{Z}$ errors. 
Then, for each logical operator, there exist three distinct ways to distribute two errors among the three available locations, as given by $\binom{3}{2}=3$. 
Next, we need to take into account that also $\M{Y}$ errors contain both $\M{X}$ and $\M{Z}$.
When the logical error under examination is $\M{X}_L$ ($\M{Z}_L$), the error patterns that could generate it are in the form of $\M{X}\M{X}$ ($\M{Z}\M{Z}$), $\M{X}\M{Y}$ ($\M{Z}\M{Y}$), 
and $\M{Y}\M{Y}$. 
This results in a total of $\binom{2}{0} + \binom{2}{1} + \binom{2}{2} = 4$ combinations.

Moving to logical operators of weight $w=4$, from~\eqref{eq:LzCyl}, we observe that we have to search for $L_4 = 18$ logical operators.
Among them, twelve are composed by $\M{X}\M{X}\M{X}\M{X}$, and the remaining six are in the form $\M{Y}\M{Y}\M{X}\M{Z}$. 
Regarding the $\M{Y}\M{Y}\M{X}\M{Z}$ logical operators, we have a similar behaviour observed for the surface code, i.e.,  after \ac{MWPM} decoding, we are always left with a logical operator with three $\M{Z}$. 
Since all the possible errors arising from this kind of logical operators have been already accounted for, we discard them. 
In other words, these logical operators cannot be triggered when only two errors are introduced by the channel and the decoder chooses for the minimum distance error pattern.
To help the visualization of such patterns let us make another example. 
Consider the $\M{X}_1\M{Y}_{7}\M{Y}_{10}\M{Z}_{13}$ operator. 
Focusing on channel errors of weight $j=2$, this logical operator could be triggered by a $\M{Y}_7\M{Y}_{10}$, since the \ac{MWPM} decoder applies an $\M{X}_1$ and a $\M{Z}_{13}$. 
However, the resulting $\M{X}_1\M{X}_{7}\M{X}_{10}$ is a stabilizer generator, and we are left with the logical error $\M{Z}_L=\M{Z}_7\M{Z}_{10}\M{Z}_{13}$.
Since we have already accounted for the case in which $\M{Z}_L=\M{Z}_7\M{Z}_{10}\M{Z}_{13}$ is generated by $\M{Y}_7\M{Y}_{10}$ when considering logical operators of weight $w=3$, we correctly discard it.

Regarding the twelve $\M{X}\M{X}\M{X}\M{X}$ logical operators, we have already observed three of them in Fig.~\ref{fig:SurfVsCylExamples}d.
In fact, these error patterns can be derived enumerating all possible paths traversing from the inside to the outside boundary of the lattice, performing a single turn.
This lead to a total of four paths starting from the same inner qubit, that makes twelve when considering that we have three inner qubits (i.e., qubit $7$, $10$, and $13$).
Then, we observe that there is a symmetry in the structure of these logical operators, and for this reason we can focus on the pair $\M{X}_L=\M{X}_{6}\M{X}_{9}\M{X}_{13}\M{X}_{14}$ and $\M{X}_L=\M{X}_{5}\M{X}_{8}\M{X}_{9}\M{X}_{13}$ to extrapolate a general behaviour.
Since they are composed only by Pauli $\M{X}$ errors, we have $\binom{4}{2}=6$ possible error patterns for each logical operator in which the $j=2$ channel errors could have occurred. 
In general, when considering the pair, we have twelve possible error locations: $i)$ two of them are always decoded because the decoder applies a correction leading to a stabilizer; $ii)$ two of them turn on the same ancilla of an operator of weight $w = 3$; $iii)$ four of them sharing in pair the same syndromes; $iv)$ four of them sharing the same syndrome.
To aid visualization, let us follow-up with our example.
Considering $\M{X}_L=\M{X}_6\M{X}_{9}\M{X}_{13}\M{X}_{14}$, the error pattern described by $i)$ is $\M{X}_{6}\M{X}_{9}$, turning on ancilla $A_{15}$.
In this case, the decoder applies $\M{X}_{15}$ and it leads to a stabilizer (i.e., successful decoding).
The error pattern described by $ii)$ is the combination $\M{X}_{13}\M{X}_{14}$ which could be ignored because it generates the logical operator $\M{X}_L=\M{X}_{13}\M{X}_{14}\M{X}_{15}$ of weight $w=3$.
The error pattern pair described by $iii)$ is composed by $\M{X}_{6}\M{X}_{13}$ and $\M{X}_{9}\M{X}_{14}$, producing the same error syndrome, meaning that at least one of them could be corrected while the other inevitably fails.
Lastly, regarding $iv)$, we have that the two operators $\M{X}_L=\M{X}_6\M{X}_{9}\M{X}_{13}\M{X}_{14}$ and $\M{X}_L=\M{X}_{5}\M{X}_{8}\M{X}_{9}\M{X}_{13}$ share $\M{X}_{9}\M{X}_{13}$, resulting in the same syndrome if this error pattern occurs.
Moreover, the errors $\M{X}_{6}\M{X}_{14}$, on the first logical operator, and $\M{X}_{5}\M{X}_{8}$, on the second logical operator, also give the same syndrome.
This means that only one of the total four possible error patterns (note that two are actually the same pattern) has to be counted.

Therefore, we conclude that each of the $6$ pairs of logical operators with $w=4$ can be caused only by $3$ error patterns (two due to $iii)$ and one due to $iv)$). The obtained $\beta_2$, over a depolarizing channel, is  
\begin{align}
\label{eq:computeBeta2c}
    \beta_2 = 1 - \frac{6\cdot3\cdot4 + 6\cdot3\cdot4 }{\binom{15}{2}\cdot 3^2} = 0.85.
\end{align}

%
%
	\label{Fig:cylinder_logicals}
%
Moving to a phase flip channel, we need to count only the patterns in the form $\M{Z}\M{Z}\M{Z}\M{Z}$. 
Since the cylindrical code has no logical operators of such a type with weight $w=4$, only the three $\M{Z}\M{Z}\M{Z}$ operators of weight $w=3$ remain.
Furthermore, for each logical operator, there exist three distinct ways to distribute two errors among the three available locations. 
This results, for the $[[15,1,3]]$ cylindrical code over a phase flip channel, in
\begin{align}
\label{eq:computeBeta2p}
    \beta_2 = 1 - \frac{3\cdot3 }{\binom{15}{2}} = 0.91\,.
\end{align} 

\subsection{M\"{o}bius Codes: Performance Analysis}
\label{subsec:MobPerfAnalysis}

Employing the procedure described above, we compute the undetectable error \ac{WE} for the $[[15,1,3]]$ M\"{o}bius code, resulting in
\begin{align}
    \nonumber L(z) &= 4z^3 + 18z^4 + 60z^5 + 220z^6 + 666z^7 + 1836z^8 \\ \nonumber &+ 4288z^9 + 7968z^{10} + 11280z^{11} + 11378z^{12} \\ &+ 7668z^{13} + 3156z^{14} + 610z^{15}. 
\end{align}
Differntly from the cylindrical code we highlight that only four logical operators with a weight of $w=3$ are present. 
Taking as example Fig.~\ref{Fig:cylinder_surface}c, the three $\M{X}\M{X}\M{X}$ logical operators are the same of the cylindrical code in Fig.~\ref{Fig:cylinder_surface}b. 
On the other hand, only one  $\M{Z}\M{Z}\M{Z}$ logical operator of the cylindrical code is still present, i.e., $\M{Z}_{8}\M{Z}_{11}\M{Z}_{14}$. 
From this, we note that, if again the logical operators of weight $w = 4$ have no $\M{Z}\M{Z}\M{Z}\M{Z}$ logical operators, this could greatly impact the performance of the code over asymmetric channels. 

Regarding logical operators of weight $w=4$, we have that all $L_4 = 18$ of them are in the form $\M{X}\M{X}\M{X}\M{X}$ Pauli operators. 
Within them, there are twelve operators that cross the lattice vertically (e.g., $\M{X}_{6}\M{X}_{9}\M{X}_{13}\M{X}_{14}$) and six that traverse it horizontally (e.g., $\M{X}_{2}\M{X}_{3}\M{X}_{6}\M{X}_{14}$). 
We note that a vertical logical operator can arise solely from one specific $\M{X}\M{X}$ error pattern among the $\binom{4}{2}$ possible combinations. 
Indeed, let us consider the  $\M{X}_{6}\M{X}_{9}\M{X}_{13}\M{X}_{14}$ logical operator.
Since this code is degenerate, the channel errors $\M{X}_{6}\M{X}_{9}$ and $\M{X}_{9}\M{X}_{13}$ are corrected by a \ac{MWPM} decoder realizing the stabilizers $\M{X}_{6}\M{X}_{9}\M{X}_{15}$ and $\M{X}_{2}\M{X}_{3}\M{X}_{9}\M{X}_{13}$, respectively. 
To be fair, here the error pattern $\M{X}_{9}\M{X}_{13}$ turns on the ancillas $A_{11}$ and $A_{14}$, and we are assuming that the \ac{MWPM} decide for the correction $\M{X}_{2}\M{X}_{3}$. 
Furthermore, the occurrence of the error $\M{X}_{13}\M{X}_{14}$ results in the logical operator $\M{X}_{13}\M{X}_{14}\M{X}_{15}$, a contribution that is taken into consideration when analyzing operators with a weight of $w = 3$. 
The channel error $\M{X}_{6}\M{X}_{14}$, turning on the ancillas $A_{11}$ and $A_{14}$, causes the horizontal logical operator $\M{X}_{2}\M{X}_{4}\M{X}_{6}\M{X}_{14}$ due to our previous assumption. 
The horizontal operators will be included later in the counting, therefore, we can discard it for this counting.
Moreover, the error patterns $\M{X}_{6}\M{X}_{13}$ and $\M{X}_{9}\M{X}_{14}$ give rise to identical syndromes, meaning that the \ac{MWPM} decoder consistently corrects one of these patterns, while the other induces the logical operator. 
In conclusion, each of the  twelve vertical logical operators can be caused only by a single error pattern of weight $t + 1$. 
Let us shift our attention to the horizontal logical operators, such as $\M{X}_2\M{X}_{3}\M{X}_{6}\M{X}_{14}$.
In this scenario, among the $\binom{4}{2} = 6$ error patterns with a weight of $j=2$, three pairs of error patterns result in identical syndromes. 
The decoder successfully corrects half of these patterns, i.e., for each of the six horizontal logical operator there are $\binom{4}{2}/2 = 3$ pattern that trigger an error. 
According to this reasoning, the value of $\beta_2$, for the $[[15,1,3]]$ M\"{o}bius code over depolarizing, we have
\begin{align}
\label{eq:computeBeta2c}
    \beta_2 = 1 - \frac{4\cdot3\cdot4 + 6\cdot3\cdot4 + 12 \cdot 4 }{\binom{15}{2}\cdot 3^2} = 0.82.
\end{align}
Since no logical operators in the form $\M{Z}\M{Z}\M{Z}\M{Z}$ exists and only one in the form $\M{Z}\M{Z}\M{Z}$ is present, for the $[[15,1,3]]$ M\"{o}bius code over phase flip channel, we have
\begin{align}
\label{eq:computeBeta2pMobius}
    \beta_2 = 1 - \frac{1\cdot3 }{\binom{15}{2}} = 0.97.
\end{align}
This shows us that M\"{o}bius codes, employing the same number of qubits of cylindrical codes, exhibit superior performance over asymmetric channels.

\subsection{Analytical Upper Bounds}
\label{subsec:UpperBounds}

In this section, we provide an upper bound on the logical qubit error rate for both the cylindrical and the M\"{o}bius codes, without limiting the code distance $d$.

In order to derive a bound valid for any $d$, we require a closed-form expression for the $L(z)$ coefficients related to the logical operators of weight $2t + 1$ and $2t + 2$.
Observing the topological structure of the cylindrical code, we have $2d$ logical operators of weight $2t + 1$: $d$ $\M{Z}$ logical operators crossing horizontally, and $d$ $\M{X}$ logical operators crossing vertically across the lattice.
Two examples of these are depicted in Fig.~\ref{Fig:cylinder_surface}b.
Moving to logical operators of weight $2t + 2$, we search for the ones composed only by $\M{X}$.
These logical operators traverse the lattice from the outer ring of the cylindrical structure to the inner ring, performing a single turn along the path.
An example of this, is given by $\M{X}_L = \M{X}_{5}\M{X}_8\M{X}_9\M{X}_{13}$ in Fig.~\ref{Fig:cylinder_surface}b.
In general, there are $d$ different starting points (i.e., qubits in the outer ring) and $d-1$ possible locations where the turn can occur, either to the right or to the left.
Then, the logical operators composed only by $\M{X}$ Pauli are $2d(d-1)$.
Additionally, there are $2d$ logical operators in the form $\M{Y}\M{Y}\M{X}\M{Z}\dots\M{Z}$, with $2t-1$ Pauli $\M{Z}$.
In conclusion, we have that $L_{2t + 1} = 2d$ and $L_{2t + 2} = 2d^2$.

Regarding the M\"{o}bius code, we have $d + 1$ logical operators of weight $2t + 1$: one $\M{Z}_L$ operator crossing horizontally, and $d$ $\M{X}_L$ operators crossing vertically across the lattice (e.g., see Fig.~\ref{Fig:cylinder_surface}c).
Regarding $\M{X}_L$ operators of weight $2t + 2$, the ones traversing horizontally the lattice are the same of the cylindrical code, i.e., $2d(d-1)$.
On the other hand, there are $d(d-1)$ $\M{X}_L$ operators of weight $2t + 2$ crossing horizontally the lattice.
Indeed, we have $d-1$ starting points (i.e., qubits connecting the outer and the inner rings) and $d$ possible locations where the turn can occur.
In conclusion, we have $L_{2t + 1} = d + 1$ and $L_{2t + 2} = 3d(d-1)$.

Furthermore, we define as $L^{X}_{w}$ the total number of $\M{X}_L$ operators composed only by $\M{X}$ operators of weight $w$.
Similarly, $L^{Z}_{w}$ is defined as the total number of $\M{Z}_L$ operators composed solely of $\M{Z}$ operators of weight $w$.
Therefore, for the cylindrical codes $L^{X}_{2t +1} = L^{Z}_{2t +1} = d$ and $L^{X}_{2t +2} = 2d(d- 1)$, while, for the M\"{o}bius codes $L^{X}_{2t +1} = d$, $L^{Z}_{2t +1} = 1$, and $L^{X}_{2t +2} = 3d(d- 1)$. 
For both codes we have that $L^{Z}_{2t +2} = 0$.


\begin{theorem}
\label{th:Bound}
    The value of $\beta_{t+1}$, for a cylindrical or M\"{o}bius code of distance $d = 2t + 1$, can be upper bounded as
    \begin{align}
    \label{eq:bound}
        \beta_{t+1} &\geq 1 - \frac{\binom{2t + 1}{t +1} L^{Z}_{2t +1}}{\binom{n}{t +1}} \left(\frac{\pz + \py}{p}\right)^{t +1} \notag \\
        &- \frac{\binom{2t + 1}{t +1} L^{X}_{2t +1}  + \binom{2t+2}{t+1} L^{X}_{2t +2} / 2 }{\binom{n}{t +1}} \left(\frac{\px + \py}{p}\right)^{t +1}
    \end{align}   
\end{theorem}

\begin{proof}
    Cylindrical and M\"{o}bius codes are \ac{CSS} codes, therefore, $t + 1$ errors comprising both $\M{X}$ and $\M{Z}$ Pauli operators are always corrected.
    Hence, $t + 1$ errors can cause a logical operator of weight $2t + 1$ or $2t + 2$ only if they are made of either $\M{X}$ and $\M{Y}$, or $\M{Z}$ and $\M{Y}$ Pauli operators. 
    As a consequence of this, we have that each error pattern composed by $\M{X}$ and $\M{Y}$, which triggers an $\M{X}_L$, has to be weighted by $(\px + \py)^{t + 1}$.
    For $\M{Z}_L$ operators, the weighting is $(\pz + \py)^{t + 1}$.
    Then, consider that each logical operator of weight $2t + 1$, can be caused by $\binom{2t + 1}{t + 1}$ possible error patterns of weight $t + 1$.
    Similarly, for each logical operator of weight $2t + 2$, there are $\binom{2t + 2}{t + 1}$ possible error patterns of weight $t + 1$.
    However, due to the \ac{MWPM} decoding, there are always a pair of different error patterns of weight $t+ 1$ causing the same syndrome when occurring over the same logical operator of weight $2t + 2$. 
    As an example, referring to Fig.~\ref{Fig:cylinder_surface}b, the errors $\M{X}_{5}\M{X}_{15}$ and $\M{X}_{7}\M{X}_{14}$ result in the same syndrome.
    As a consequence, only half of the $\binom{2t + 2}{t + 1}$ combinations should be counted.
    In conclusion, the enumerators in \eqref{eq:bound} are the total number of faulty error patterns comprising $t+1$ Pauli operators, weighted by their probability of occurrence.
    On the other hand, the common denominator is the total number of error patterns of weight $t+1$ that can occur over $n$ qubits, $\binom{n}{t + 1}$, weighted by its probability.
    We remark that this is an upper bound since we are not taking into account the degeneracy of the code.
    Specifically, we are not considering the overlapping of logical operators, which reduces the total number of possible error patterns given by the binomials~\cite{ForValChi:23}.
\end{proof}


\begin{corollary}
   \label{cor:AsymmetrycChannel}
   For a cylindrical or M\"{o}bius code of distance $d = 2t + 1$ over an asymmetric channel with asymmetry parameter $A = 2\pz / (p - \pz)$ and $\px = \py$, the logical error rate can be upper bounded as

   \begin{align}
   \label{eq:boundCor}
       p_\mathrm{L} \leq &\frac{ (A+1)^{t+1} \binom{2t + 1}{t +1} L^{Z}_{2t +1} }{(A+2)^{t+1}} \, p^{t+1} \nonumber \\ 
       &+ \frac{ 2^{t+1} \left[ \binom{2t + 1}{t +1} L^{X}_{2t +1}  + \binom{2t+2}{t+1} L^{X}_{2t +2} / 2 \right]}{(A+2)^{t+1}} \, p^{t+1} \,.
   \end{align}

\end{corollary}

\section{Numerical Results}\label{sec:NumRes}

\begin{table*}[t]
    \centering
    \setlength{\tabcolsep}{3pt}
    \caption{Fraction of non-correctable error patterns per error class of cylindrical codes using \acl{MWPM}.}
    \label{tab:Err}
    \small
    \begin{tabular}{lC{1.3cm}C{1.3cm}C{1.3cm}C{1.3cm}C{1.3cm}C{1.3cm}C{1.3cm}C{1.3cm}C{1.3cm}C{1.3cm}}
        \toprule
        \rowcolor[gray]{.95}
        \textbf{Code} & $\M{X}\M{X}$ & $\M{X}\M{Z}$ & $\M{X}\M{Y}$ & $\M{Z}\M{Z}$ & $\M{Z}\M{Y}$ & $\M{Y}\M{Y}$ & & & &\\
        \midrule
        $[[15,1,3]]$ Cyl. & $0.257$ & $0$ & $0.257$  & $0.086$  & $0.086$ & $0.343$ \\
        $[[15,1,3]]$ M\"{o}b. & $0.371$ & $0$ & $0.371$  & $0.029$  & $0.029$ & $0.400$ \\
        $[[25,1,3/5]]$ Cyl. & $0.150$ & $0$ & $0.150$ & $0$ & $0$ & $0.150$ \\
        $[[25,1,3/5]]$ M\"{o}b. & $0.150$ & $0$ & $0.150$ & $0$ & $0$ & $0.150$ \\
        \midrule 
        \rowcolor[gray]{.95}
        \textbf{Code} & $\M{X}\M{X}\M{X}$ & $\M{X}\M{X}\M{Z}$ & $\M{X}\M{X}\M{Y}$ & $\M{X}\M{Z}\M{Z}$ & $\M{X}\M{Z}\M{Y}$ & $\M{X}\M{Y}\M{Y}$ & $\M{Z}\M{Z}\M{Z}$ & $\M{Z}\M{Z}\M{Y}$ & $\M{Z}\M{Y}\M{Y}$ & $\M{Y}\M{Y}\M{Y}$\\
        \midrule
        $[[25,1,3/5]]$ Cyl. & $0.384$ & $0.150$ & $0.384$ & $0$ & $0.150$ & $0.384$ & $0.013$ & $0.013$ & $0.163$ & $0.397$ \\
        $[[25,1,3/5]]$ M\"{o}b. & $0.396$ & $0.150 $ & $0.396$ & $0$ & $0.150$ & $0.396 $ & $0.004$ & $0.004$ & $0.154$ & $0.401$ \\
        $[[45,1,5]]$ Cyl. & $0.019$ & $0$ & $0.019$ & $0$ & $0$ & $0.019$ & $0.004$ & $0.004$ & $0.004$ & $0.023$ \\
        $[[45,1,5]]$ M\"{o}b. & $0.025 $ & $0$ & $0.025  $ & $0$ & $0.150$ & $0.0251 $ & $7\cdot 10^{-4}$ & $7\cdot 10^{-4}$ & $7\cdot 10^{-4}$ & $0.401  $ \\
        \midrule 
     \end{tabular}
     \begin{tabular}{lC{2cm}C{2cm}C{2cm}C{2cm}C{2cm}C{2cm}C{2cm}}
         & & \multicolumn{4}{c} {$1 - \beta_2(A)$} \\
        \cmidrule{3-6} 
        \rowcolor[gray]{.95}
        \textbf{Code}& & $A=1$ & $A=10$ & $A=100$ & $A\to \infty$ & & \\
        \midrule
       $[[13,1,3]]$ & & $0.24$ &$0.233$ &$0.265$ & $0.270$ & & \\
        $[[15,1,3]]$ Cyl. & & $0.15$ & $0.080$ & $0.084$& $0.091$ & & \\
        $[[15,1,3]]$ M\"{o}b. & & $0.18$ & $0.034$ & $0.028$& $0.029$ & & \\
        \bottomrule
    \end{tabular}
\end{table*}

In this section we numerically evaluate the performance of cylindrical codes with \ac{MWPM} decoding via Monte Carlo simulations and we provide a comparison with surface codes. 
All numerical simulations are performed by running the decoder until a minimum of 100 errors are reached, ensuring reliable results.
In doing so, we exploit Lemon C++ library \cite{DezBalJut:11}, which provides an implementation of graphs and networks algorithms.
Moreover, In Tab.~\ref{tab:Err} we report for some cylindrical codes and M\"{o}bius codes of interest the percentage of non-correctable errors for each error class $f_j(i,\ell)$, which we have evaluated by exhaustive search with a \ac{MWPM} decoder. 
For instance, in the case of the $[[15,1,3]]$ cylindrical code, it results $f_2(0,2) =~0.257$. 
Exploiting these tabular values, we can easily write analytical expressions for the code performance. 
This result can be used to analyze and design complex systems without implementing the decoder. 
As an example, for small $p$, the logical error rate of a $[[15,1,3]]$ cylindrical code tends to 
\begin{align}\label{eq:RhoL913XZZX}
    p_\mathrm{L} \approx \frac{0.086A^2 + 0.172A + 1.114}{(A + 2)^2} \, \binom{15}{2} p^2 \,.
\end{align}
\begin{figure}[t]
	\centering
	\resizebox{\columnwidth}{!}{ 
%
%
\definecolor{mycolor1}{rgb}{0.00000,0.44700,0.74100}%
\definecolor{mycolor2}{rgb}{0.85000,0.32500,0.09800}%
\definecolor{mycolor3}{rgb}{0.92900,0.69400,0.12500}%
\definecolor{mycolor4}{rgb}{0.49400,0.18400,0.55600}%
\definecolor{mycolor5}{rgb}{0.46600,0.67400,0.18800}%
\begin{tikzpicture}

\begin{axis}[%
name = match,
width=4.5in,
height=3.5in,
at={(0in,0in)},
scale only axis,
xmode=log,
xmin=0.005,
xmax=0.14,
xminorticks=true,
xlabel={ Physical Qubit Error Probability $p$},
xlabel style={font=\color{white!15!black}, font = \Large},
ymode=log,
ymin=0.0001,
ymax=1,
yminorticks=true,
ylabel={ Logical Qubit Error Probability $p_L$},
ylabel style={font=\color{white!15!black}, font = \Large},
axis background/.style={fill=white},
tick label style={black, semithick, font = \fontsize{13.5}{13.5}},
xmajorgrids,
ymajorgrids,
legend style={at={(0.03,0.97)}, anchor=north west, legend cell align=left, align=left, font = \Large,  draw=white!15!black}
]



\addplot [color=graphGreen, line width=1.4pt, dashed
]
  table[row sep=crcr]{%
0.001	1.872e-05\\
0.01021429	0.00195308980222075\\
0.01942857	0.00706622589962453\\
0.02864286	0.0153581393904981\\
0.03785714	0.0268288122768981\\
0.04707143	0.0414782694564245\\
0.05628571	0.0593064791318208\\
0.0655	0.08031348\\
0.07471429	0.104499254437421\\
0.08392857	0.131863771021225\\
0.09314286	0.162407089147298\\
0.10235714	0.196129142520098\\
0.11157143	0.233030004334825\\
0.12078571	0.273109594496621\\
0.13	0.316368\\
0.139	0.36168912\\
0.1482	0.4111518528\\
0.1574	0.4637835072\\
}; \label{plot:[[13,1,3]] depo}

\addplot [color=brightRed, line width=1.4pt, dashed
]
  table[row sep=crcr]{%
0.001	1.575e-05\\
0.01021429	0.00164322459321457\\
0.01942857	0.00594514198285718\\
0.02864286	0.0129215115064287\\
0.03785714	0.0225723180214287\\
0.04707143	0.0348975824753572\\
0.05628571	0.0498972781157146\\
0.0655	0.0675714375\\
0.07471429	0.0879200458007146\\
0.08392857	0.110943076580357\\
0.09314286	0.136640579811429\\
0.10235714	0.165012499716429\\
0.11157143	0.196058897877857\\
0.12078571	0.229779706908215\\
0.13	0.266175\\
0.139	0.30430575\\
0.1482	0.34592103\\
0.1574	0.39020247\\
}; \label{plot:[[15,1,3]] depo}

\addplot [color=graphGreen, line width=1.4pt 
]
  table[row sep=crcr]{%
0.001	2.106e-05\\
0.01021429	0.00219722602749835\\
0.01942857	0.00794950413707759\\
0.02864286	0.0172779068143104\\
0.03785714	0.0301824138115104\\
0.04707143	0.0466630531384776\\
0.05628571	0.0667197890232984\\
0.0655	0.090352665\\
0.07471429	0.117561661242098\\
0.08392857	0.148346742398878\\
0.09314286	0.18270797529071\\
0.10235714	0.22064528533511\\
0.11157143	0.262158754876678\\
0.12078571	0.307248293808698\\
0.13	0.355914\\
0.139	0.40690026\\
0.1482	0.4625458344\\
0.1574	0.5217564456\\
};\label{plot:[[13,1,3]] p flip}

\addplot [color=brightRed, line width=1.4pt 
]
  table[row sep=crcr]{%
0.001	8.925e-06\\
0.01021429	0.000931160602821592\\
0.01942857	0.00336891379028573\\
0.02864286	0.00732218985364293\\
0.03785714	0.0127909802121429\\
0.04707143	0.0197752967360357\\
0.05628571	0.0282751242655716\\
0.0655	0.03829048125\\
0.07471429	0.0498213592870716\\
0.08392857	0.0628677433955357\\
0.09314286	0.0774296618931429\\
0.10235714	0.0935070831726429\\
0.11157143	0.111100042130786\\
0.12078571	0.130208500581322\\
0.13	0.1508325\\
0.139	0.172439925\\
0.1482	0.196021917\\
0.1574	0.221114733\\
}; \label{plot:[[15,1,3]] p flip}


\addplot [only marks, color=brightRed, line width=1.3pt, mark=o, mark size = 2.5pt, mark options={solid, brightRed}]
  table[row sep=crcr]{%
0.01	0.001585\\
0.02	0.00596\\
0.05	0.03517\\
0.08	0.08294\\
0.1	0.12049\\
0.12	0.16245\\
}; \label{plot:simdep}

\addplot [only marks, color=brightRed, line width=1.3pt, mark=square, mark size = 2.5pt, mark options={solid, brightRed}]
  table[row sep=crcr]{%
0.01	0.00091\\
0.02	0.00364\\
0.05	0.02593\\
0.08	0.06891\\
0.1	0.10512\\
0.12	0.14786\\
}; \label{plot:simph}



\end{axis}

\node[draw,fill=white,inner sep=1.5pt,above left=0.5em] at (match.south east){
  {\renewcommand{\arraystretch}{1.2}
    \begin{tabular}{lcc}
        & $A = 1$ & $A \to \infty$  \\
       \text{$[[13,1,3]]$} & \ref{plot:[[13,1,3]] depo} & \ref{plot:[[13,1,3]] p flip}  \\
        \text{$[[15,1,3]]$ Cyl.} & \ref{plot:[[15,1,3]] depo} & \ref{plot:[[15,1,3]] p flip}  \\
        \text{$[[15,1,3]]$ Cyl. Sim.} & \ref{plot:simdep} & \ref{plot:simph}  \\
  \end{tabular}
}};

\end{tikzpicture}%
	} 
	\caption{Logical error rate over physical error rate, $[[13,1,3]]$ surface code and $[[15,1,3]]$ cylindrical code with \ac{MWPM} decoder. Depolarizing and phase flip channels channel. The curves refer to the asymptotic approximations \eqref{eq:error_probWithBetaApprox}. Squares and circles correspond to \ac{MWPM} decoding. 
		\label{Fig:plot_disk_match}} 
\end{figure}
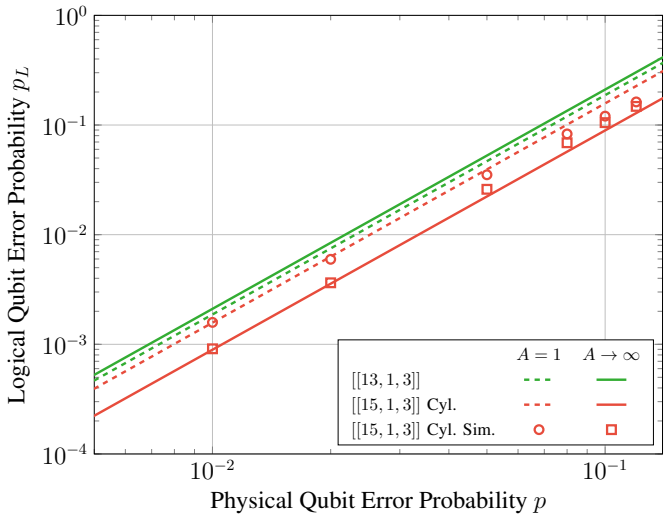

\emph{1) Numerical validation:}
Fig.~\ref{Fig:plot_disk_match} shows a comparison between the asymptotic logical error rate for the $[[13,1,3]]$ surface code and the $[[15,1,3]]$ cylindrical code over depolarizing and phase flip channels, assuming a \ac{MWPM} decoder. These error probabilities are computed using \eqref{eq:error_probWithBetaApprox} with the values of $\beta_2$ presented above. We can see that, over a symmetric channel, the cylindrical code works better but the performance are similar to the surface code. However, the logical error rate of the surface code increases with the bias of the channel. On the contrary, the ability of the $[[15,1,3]]$ cylindrical code to correct all possible $\M{Z}$ ($\M{X}$ for the dual code) logical errors of weight $w = 4$ increases its error correction capability. Indeed, in a channel with $A \to \infty$, the cylindrical code works much better than the standard surface code. 
In Fig.~\ref{Fig:plot_disk_match} we also provide a comparison with \ac{MWPM} simulations for the $[[15,1,3]]$ cylindrical code, showing that our asymptotic approximations \eqref{eq:error_probWithBetaApprox} are tight to the Monte Carlo simulations for values of physical error rate $p < 0.02$.

\emph{2) Logical error rate vs. channel asymmetry:}  
Fig.~\ref{Fig:plot_varA} shows a comparison between the asymptotic performance, computed using \eqref{eq:error_probWithBetaApprox}, of the $[[13,1,3]]$ and the $[[41,1,5]]$ surface codes, and the $[[15,1,3]]$, and the $[[45,1,5]]$ cylindrical and M\"{o}bius codes, as a function of the channel asymmetry parameter $A$. 
These curves are computed for a value of physical error rate $p = 0.01$. In the case of the $[[13,1,3]]$ and the $[[41,1,5]]$ surface codes, the logical error rate has a minimum value for $A=2.9$ and for $A = 1.8$, respectively. Moreover, increasing or decreasing the asymmetry of the channel with respect to this value, reduces the error correction capability of the same amount. Indeed, symmetric surface codes, especially those with distance $d > 3$, do not perform well over channels where one kind of Pauli error happens rather more frequently than the others. 
On the contrary, for the $[[15,1,3]]$ and the $[[45,1,5]]$ cylindrical codes, the logical error rate has a minimum for $A = 10.7$ and for $A = 5.3$, respectively. 
Specifically, for values of asymmetry higher than these minima, the logical error rate increases very slowly. 
This asses how the error correction capability of cylindrical codes is enhanced over asymmetric channels. 
The logical error rate of the $[[15,1,5]]$  M\"{o}bius code has its minimum value for $A = 78.8$, while the $[[45,1,5]]$ for $A = 14.9$. 
Moreover, these codes outperform surface codes across all values of channel asymmetry, exhibiting a significant improvement in error correction capability, i.e. exceeding an order of magnitude for moderate asymmetry values. 

To validate our theoretical findings, we compute bounds on error correction capability of cylindrical and M\"{o}bius codes, employing \eqref{eq:boundCor}. 
We observe that these upper bounds are very tight to the real logical error rate for a wide range of channel asymmetry values. 
Finally, we want to discuss the reason why the performance curves exhibit a minimum when varying the asymmetry. 
This behavior occurs because the code can successfully correct all patterns involving both $\mathbf{X}$ and $\mathbf{Z}$ errors, while there are a few patterns consisting solely of $\mathbf{Z}$ errors that remain uncorrected (see Table~\ref{tab:Err}).
More generally, this phenomenon arises in \ac{CSS} codes where $\mathbf{X}$ and $\mathbf{Z}$ errors are decoded independently. Specifically, if a code can correct up to $t$ errors, a decoder that treats $\mathbf{X}$ and $\mathbf{Z}$ errors independently will always be able to correct any pattern consisting of up to $t$ $\mathbf{X}$ errors plus $t$ $\mathbf{Z}$ errors.


%
\begin{figure}[t]
	\centering
	\resizebox{\columnwidth}{!}{ 
	    \input{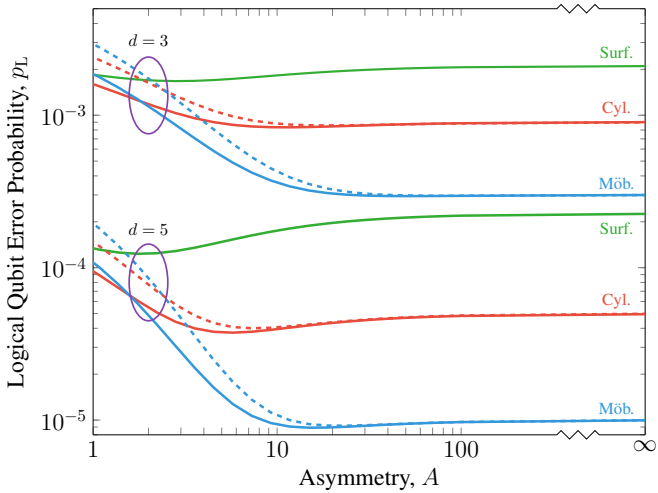} 
	} 
	\caption{Asymptotic behaviour of the logical error rate over channel asymmetry $A$ for surface and cylindrical, and M\"{o}bius codes. The logical error probabilities are computed at physical error probability $p = 0.01$. Solid line: exact \eqref{eq:error_probWithBetaApprox}. Dashed line: upper bound \eqref{eq:boundCor}.}
	\label{Fig:plot_varA}
\end{figure}

\emph{3) Analytical upper bound:}
Fig.~\ref{Fig:plot_BOUNDS} shows the analytical upper bound \eqref{eq:boundCor} for the $[[91,1,7]]$, the $[[153,1,9]]$, and the $[[231,1,11]]$ cylindrical and M\"{o}bius codes for a physical error probability $p = 0.001$.
Note that for codes of these dimensions, it becomes impractical to compute the exact asymptotic logical error rate through exhaustive search of faulty error patterns, as was feasible for smaller codes, or to simulate performance using the \ac{MWPM} decoder.
Consequently, this bound represents an important alternative for estimating the logical error probability of these codes, particularly when the channel asymmetry is significantly pronounced.

\begin{figure}[t]
	\centering
	\resizebox{\columnwidth}{!}{ 
	    \input{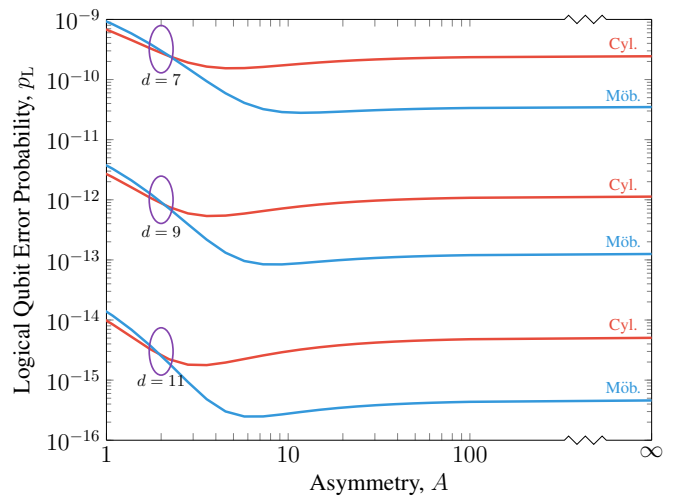} 
	} 
	\caption{Evaluating the code performance for high values of distance, using the analytical upper bound \eqref{eq:boundCor}. The logical error probabilities are computed at physical error probability $p = 0.001$.}
	\label{Fig:plot_BOUNDS}
\end{figure}

\emph{4) Logical error rate vs. physical error rate:} 
Fig.~\ref{Fig:plot_diskfinal} shows the asymptotic logical error rate of the distance $d=5$ surface, cylindrical and M\"{o}bius codes, over channels with different values of asymmetries. As anticipated, over a depolarizing channel ($A=1$), the M\"{o}bius code's error correction capability is inferior to that of the cylindrical code. Indeed, although M\"{o}bius code has only $\binom{d}{t+1}$ faulty $\M{Z}$ error patterns 
(which are part of the untwisted central row), it has a bigger amount of faulty $\M{X}$ errors of weight $t+1$ with respect to the cylindrical version.  Moreover, as the asymmetry of the channel increases, the error correction capability of cylindrical and M\"{o}bius codes is enhanced. In particular, in the asymptotic condition $A \to \infty$, the $[[45,1,5]]$ M\"{o}bius code outperforms both surface and cylindrical codes. Additionally, it's worth noting that the performance of the M\"{o}bius code over a channel with $A=10$ is nearly indistinguishable from that of a phase flip channel. This observation underscores the great advantage this structure provides in terms of reducing logical error rates when dealing with slightly asymmetric channels. 

Finally, in Fig.~\ref{Fig:plot_diskF2} displays the logical error rates for various asymmetric surface, cylindrical, and M\"{o}bius codes. Specifically, for moderately polarized channels, the performance of these asymmetric codes is quite similar. However, as the channel asymmetry increases, the M\"{o}bius code emerges as the best performer. Note that the slope of the three curves changes moving from $A= 10$ to $A \to \infty$ since these codes have distance $d=5$ over a phase flip channel.     
%
%

%
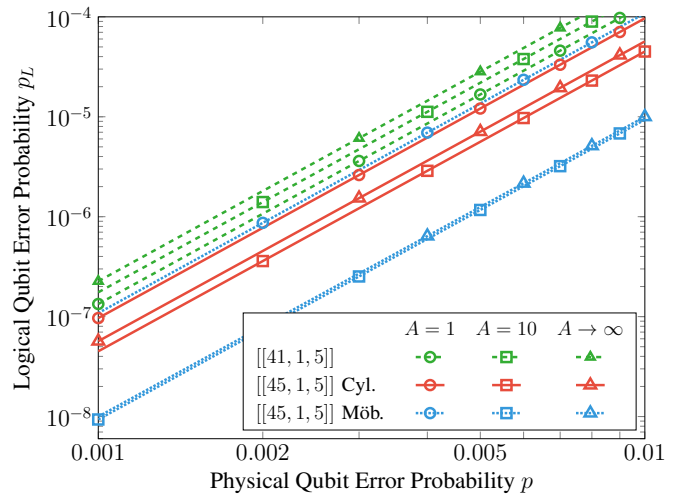
\begin{figure}[t]
	\centering
	\resizebox{0.49\textwidth}{!}{ 
%
%
\definecolor{mycolor1}{rgb}{0.00000,0.44700,0.74100}%
\definecolor{mycolor2}{rgb}{0.85000,0.32500,0.09800}%
\definecolor{mycolor3}{rgb}{0.92900,0.69400,0.12500}%
\definecolor{mycolor4}{rgb}{0.49400,0.18400,0.55600}%
\definecolor{mycolor5}{rgb}{0.46600,0.67400,0.18800}%
\definecolor{mycolor6}{rgb}{0.30100,0.74500,0.93300}%
\definecolor{mycolor7}{rgb}{0.63500,0.07800,0.18400}%
\begin{tikzpicture}

\begin{axis}[%
name = plot,
width=4.5in,
height=3.5in,
at={(0in,0in)},
scale only axis,
xmode=log,
xmin=0.001,
xmax=0.01,
xminorticks=true,
xlabel style={font=\color{white!15!black}, font = \Large},
xlabel={  Physical Qubit Error Probability $p$},
xtick = {0.001, 0.002, 0.003, 0.004, 0.005, 0.006, 0.007, 0.008, 0.009, 0.01, 0.02, 0.03, 0.04, 0.05},
xticklabels = {$0.001$, $0.002$, , , $0.005$, , , , , $0.01$, $0.02$, , , $0.05$},
tick label style={black, semithick, font=\Large},
ymode=log,
ymin=6e-09,
ymax=1e-04,
yminorticks=true,
ylabel style={font=\color{white!15!black}, font = \Large},
ylabel={ Logical Qubit Error Probability $p_L$},
axis background/.style={fill=white},
title style={font=\bfseries},
legend style={legend cell align=left, align=left, draw=white!15!black}
]

\addplot [color=graphGreen, dashed, line width=1.5pt, mark=o, mark size = 3pt, mark options={solid, graphGreen}, mark repeat = 2 ]
  table[row sep=crcr]{%
0.001	1.3376325925926e-07\\
0.002	1.07010607407408e-06\\
0.003	3.61160800000001e-06\\
0.004	8.56084859259263e-06\\
0.005	1.67204074074075e-05\\
0.006	2.88928640000001e-05\\
0.007	4.58807979259261e-05\\
0.008	6.8486788740741e-05\\
0.009	9.75134160000004e-05\\
0.01	0.00013376325925926\\
}; \label{sur:41_1}

\addplot [color=graphGreen, dashed, line width=1.5pt, mark=square, mark size = 3pt, mark options={solid, graphGreen}, mark repeat = 2, mark phase = 2]
  table[row sep=crcr]{%
0.001	1.75115175925926e-07\\
0.002	1.40092140740741e-06\\
0.003	4.72810975000001e-06\\
0.004	1.12073712592593e-05\\
0.005	2.18893969907408e-05\\
0.006	3.78248780000001e-05\\
0.007	6.00645053425928e-05\\
0.008	8.96589700740743e-05\\
0.009	0.00012765896325\\
0.01	0.000175115175925926\\
}; \label{sur:41_1 A10}

\addplot [color=graphGreen, dashed, line width=1.5pt, mark=triangle, mark size = 3pt, mark options={solid, graphGreen}, mark repeat = 2]
  table[row sep=crcr]{%
0.001	2.25991322026034e-07\\
0.002	1.80793057620827e-06\\
0.003	6.10176569470292e-06\\
0.004	1.44634446096662e-05\\
0.005	2.82489152532542e-05\\
0.006	4.88141255576234e-05\\
0.007	7.75150234549297e-05\\
0.008	0.000115707556877329\\
0.009	0.000164747673756979\\
0.01	0.000225991322026034\\
}; \label{sur:41_1 pp}

\addplot [color=brightRed, line width=1.5pt, mark=o, mark size = 3pt, mark options={solid, brightRed}, mark repeat = 2]
  table[row sep=crcr]{%
0.001	9.67022222222224e-08\\
0.002	7.73617777777779e-07\\
0.003	2.61096000000001e-06\\
0.004	6.18894222222224e-06\\
0.005	1.20877777777778e-05\\
0.006	2.08876800000001e-05\\
0.007	3.31688622222223e-05\\
0.008	4.95115377777779e-05\\
0.009	7.04959200000002e-05\\
0.01	9.67022222222224e-05\\
}; \label{cyl:45_A1}

\addplot [color=brightRed, line width=1.5pt, mark=square, mark size = 3pt, mark options={solid, brightRed}, mark repeat = 2, mark phase = 2]
  table[row sep=crcr]{%
0.001	4.49678472222222e-08\\
0.002	3.59742777777778e-07\\
0.003	1.214131875e-06\\
0.004	2.87794222222222e-06\\
0.005	5.62098090277778e-06\\
0.006	9.713055e-06\\
0.007	1.54239715972222e-05\\
0.008	2.30235377777778e-05\\
0.009	3.2781560625e-05\\
0.01	4.49678472222222e-05\\
}; \label{cyl:45_A10}

\addplot [color=brightRed, line width=1.5pt, mark=triangle, mark size = 4pt, mark options={solid, brightRed}, mark repeat = 2]
  table[row sep=crcr]{%
0.001	5.67598297205106e-08\\
0.002	4.54078637764085e-07\\
0.003	1.53251540245379e-06\\
0.004	3.63262910211268e-06\\
0.005	7.09497871506382e-06\\
0.006	1.22601232196303e-05\\
0.007	1.94686215941351e-05\\
0.008	2.90610328169014e-05\\
0.009	4.13779158662522e-05\\
0.01	5.67598297205106e-05\\
}; \label{cyl:45_pp}

\addplot [color=brightBlue, densely dotted, line width=1.5pt, mark=o, mark size=3pt, mark options={solid, brightBlue}, mark repeat = 2, mark phase = 2 ]
  table[row sep=crcr]{%
0.001	1.08489382222222e-07\\
0.002	8.67915057777777e-07\\
0.003	2.92921332e-06\\
0.004	6.94332046222222e-06\\
0.005	1.35611727777778e-05\\
0.006	2.343370656e-05\\
0.007	3.72118581022222e-05\\
0.008	5.55465636977777e-05\\
0.009	7.908875964e-05\\
0.01	0.000108489382222222\\
}; \label{mob:45_A1}

\addplot [color=brightBlue, densely dotted, line width=1.5pt, mark=square, mark size = 3pt, mark options={solid, brightBlue}, mark repeat = 2]
  table[row sep=crcr]{%
0.001	9.3435565972218e-09\\
0.002	7.47484527777744e-08\\
0.003	2.52276028124989e-07\\
0.004	5.97987622222195e-07\\
0.005	1.16794457465272e-06\\
0.006	2.01820822499991e-06\\
0.007	3.20483991284708e-06\\
0.008	4.78390097777756e-06\\
0.009	6.81145275937469e-06\\
0.01	9.3435565972218e-06\\
}; \label{mob:45_A10}

\addplot [color=brightBlue, densely dotted, line width=1.5pt, mark=triangle, mark size = 4pt, mark options={solid, brightBlue}, mark repeat = 2, mark phase = 2]
  table[row sep=crcr]{%
0.001	9.98973003081006e-09\\
0.002	7.99178402464805e-08\\
0.003	2.69722710831872e-07\\
0.004	6.39342721971844e-07\\
0.005	1.24871625385126e-06\\
0.006	2.15778168665497e-06\\
0.007	3.42647740056785e-06\\
0.008	5.11474177577475e-06\\
0.009	7.28251319246054e-06\\
0.01	9.98973003081007e-06\\
}; \label{mob:45_pp}

\end{axis}

\node[draw,fill=white,inner sep=1.5pt, above left=0.5em] at (plot.south east){
\large
  {\renewcommand{\arraystretch}{1.2}
    \begin{tabular}{lccc}
    & \textbf{$A=1$} & \textbf{$A=10$} & \textbf{$A \to \infty$} \\
    \text{$[[41,1,5]]$} & \ref{sur:41_1} & $\ref{sur:41_1 A10}$ & $\ref{sur:41_1 pp}$  \\
    \text{$[[45,1,5]]$ Cyl.}  & \ref{cyl:45_A1} & \ref{cyl:45_A10} & \ref{cyl:45_pp} \\ 
    \text{$[[45,1,5]]$ M\"{o}b.}  & \ref{mob:45_A1} & \ref{mob:45_A10} & \ref{mob:45_pp} \\
  \end{tabular}
}};

\end{tikzpicture}%
	} 
	\caption{Performance of  topological codes with symmetric structure. $[[41,1,5]]$ surface code, $[[45,1,5]]$ cylindrical code, and $[[45,1,5]]$ M\"{o}bius code. Asymmetric errors, with asymmetry $A=1$, $A=10$, $A \to \infty$.
		\label{Fig:plot_diskfinal}}
\end{figure}
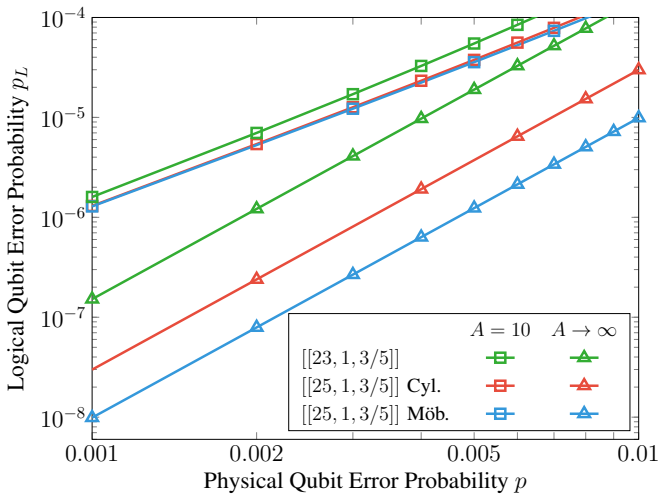
\begin{figure}[t]
	\centering
	\resizebox{0.49\textwidth}{!}{ 
%
%
\definecolor{mycolor1}{rgb}{0.00000,0.44700,0.74100}%
\definecolor{mycolor2}{rgb}{0.85000,0.32500,0.09800}%
\definecolor{mycolor3}{rgb}{0.92900,0.69400,0.12500}%
\definecolor{mycolor4}{rgb}{0.49400,0.18400,0.55600}%
\definecolor{mycolor5}{rgb}{0.46600,0.67400,0.18800}%
\definecolor{mycolor6}{rgb}{0.30100,0.74500,0.93300}%
\definecolor{mycolor7}{rgb}{0.63500,0.07800,0.18400}%
\begin{tikzpicture}

\begin{axis}[%
name = plot,
width=4.5in,
height=3.5in,
at={(0in,0in)},
scale only axis,
xmode=log,
xmin=0.001,
xmax=0.01,
xminorticks=true,
xlabel style={font=\color{white!15!black}, font = \Large},
xlabel={ Physical Qubit Error Probability $p$},
xtick = {0.001, 0.002, 0.003, 0.004, 0.005, 0.006, 0.007, 0.008, 0.009, 0.01, 0.02, 0.03, 0.04, 0.05},
xticklabels = {$0.001$, $0.002$, , , $0.005$, , , , , $0.01$, $0.02$, , , $0.05$},
tick label style={black, semithick, font=\Large},
ymode=log,
ymin=6e-09,
ymax=1e-04,
yminorticks=true,
ylabel style={font=\color{white!15!black}, font = \Large},
ylabel={ Logical Qubit Error Probability $p_L$},
axis background/.style={fill=white},
title style={font=\bfseries},
legend style={legend cell align=left, align=left, draw=white!15!black}
]
\addplot [color=brightRed, line width=1.5pt, mark=square, mark size = 3pt, mark options={solid, brightRed}]
  table[row sep=crcr]{%
0.001	1.30107783564814e-06\\
0.002	5.40862268518517e-06\\
0.003	1.26291015625e-05\\
0.004	2.32689814814814e-05\\
0.005	3.76347294560184e-05\\
0.006	5.60328124999999e-05\\
0.007	7.87696976273146e-05\\
0.008	0.000106151851851852\\
0.009	0.0001384857421875\\
0.01	0.000176077835648148\\
}; \label{cyl:25_10}

\addplot [color=brightRed, line width=1.5pt, mark=triangle, mark size = 4pt, mark options={solid, brightRed}, mark repeat = 2, mark phase = 2]
  table[row sep=crcr]{%
0.001	2.98999104843985e-08\\
0.002	2.39199283155231e-07\\
0.003	8.07297579838952e-07\\
0.004	1.91359426236202e-06\\
0.005	3.73748879255088e-06\\
0.006	6.458380632232e-06\\
0.007	1.02556692432318e-05\\
0.008	1.53087540873768e-05\\
0.009	2.17970346264934e-05\\
0.01	2.98999103224081e-05\\
}; \label{cyl:25_pp}

\addplot [color=brightBlue, line width=1.5pt, mark=square, mark size = 3pt, mark options={solid, brightBlue}, mark repeat = 2 ]
  table[row sep=crcr]{%
0.001	1.28579388310185e-06\\
0.002	5.2863510648148e-06\\
0.003	1.221643484375e-05\\
0.004	2.22908085185185e-05\\
0.005	3.57242353877314e-05\\
0.006	5.27314787499999e-05\\
0.007	7.3527301903935e-05\\
0.008	9.83264681481479e-05\\
0.009	0.00012734374078125\\
0.01	0.000160793883101852\\
}; \label{mob:25_10}

\addplot [color=brightBlue, line width=1.5pt, mark=triangle, mark size = 4pt, mark options={solid, brightBlue} ]
  table[row sep=crcr]{%
0.001	9.88997051421826e-09\\
0.002	7.91197633937887e-08\\
0.003	2.67029200644085e-07\\
0.004	6.3295810427048e-07\\
0.005	1.23624629627835e-06\\
0.006	2.13623359867306e-06\\
0.007	3.39225983345999e-06\\
0.008	5.06366482264452e-06\\
0.009	7.20978838823201e-06\\
0.01	9.88997035222784e-06\\
}; \label{mob:25_pp}

\addplot [color=graphGreen, line width=1.5pt, mark = square, mark size = 3pt, mark=square, mark options={solid, graphGreen}]
  table[row sep=crcr]{%
0.001	1.59912065972223e-06\\
0.002	6.9929652777778e-06\\
0.003	1.70762578125001e-05\\
0.004	3.27437222222223e-05\\
0.005	5.48900824652779e-05\\
0.006	8.44100625000003e-05\\
0.007	0.000122198386284723\\
0.008	0.000169149777777778\\
0.009	0.000226158960937501\\
0.01	0.000294120659722223\\
}; \label{sur:23_A10}

\addplot [color=graphGreen, line width=1.5pt, mark=triangle, mark size = 4pt, mark options={solid, graphGreen}]
  table[row sep=crcr]{%
0.001	1.51799544814934e-07\\
0.002	1.21439635768427e-06\\
0.003	4.09858770624482e-06\\
0.004	9.71517085813337e-06\\
0.005	1.89749430809867e-05\\
0.006	3.27887016424418e-05\\
0.007	5.20672438101352e-05\\
0.008	7.77213668517038e-05\\
0.009	0.000110661868034785\\
0.01	0.000151799544627014\\
}; \label{sur:23_pp}

\end{axis}

\node[draw,fill=white,inner sep=1.5pt, above left=0.5em] at (plot.south east){
\large
  {\renewcommand{\arraystretch}{1.2}
    \begin{tabular}{lcc}
     & \textbf{$A=10$} & \textbf{$A \to \infty$} \\
     \textbf{$[[23,1,3/5]]$}  & \ref{sur:23_A10} & \ref{sur:23_pp} \\
    \textbf{$[[25,1,3/5]]$} Cyl. & \ref{cyl:25_10} & \ref{cyl:25_pp} \\
    \textbf{$[[25,1,3/5]]$} M\"{o}b.  & \ref{mob:25_10} & \ref{mob:25_pp} \\
    
  \end{tabular}
}};

\end{tikzpicture}%
	} 
	\caption{Performance of  topological codes with asymmetric structure. $[[23,1,3/5]]$ surface code, $[[25,1,3/5]]$ cylindrical code, and $[[25,1,3/5]]$ M\"{o}bius code. Asymmetric errors, with asymmetry $A=10$, $A \to \infty$.}
    \label{Fig:plot_diskF2}
\end{figure}

\begin{table}[t]
    \centering
    \caption{Threshold estimations of Cylindrical and M\"{o}bius codes.}
    \label{tab:thre}
    \begin{tabular}{l c c c c}
    \toprule
         & $A = 1$ & $A = 10$ & $A = \infty$ \\
    \midrule
        Cylindrical Codes & $0.14$ & $0.12$ & $0.10$ \\
         M\"{o}bius Codes & $0.14$ & $0.12$ & $0.10$ \\
    \bottomrule
    \end{tabular}
\end{table}

\emph{5) High physical error rate analysis:} 
For the sake of completeness, we also report the threshold values, a metric related to the performance in the high physical error rate regime.
In particular, Tab.~\ref{tab:thre} shows the threshold values obtained via Monte Carlo simulations for the cylindrical and M\"{o}bius codes. 
The simulations were conducted for channels with varying asymmetry values, revealing that an increase in asymmetry leads to a decrease in the threshold value.

\section{Conclusions}\label{sec:conclusions}
We have presented new topological codes, the cylindrical codes and M\"{o}bius codes. 
We have derived their code parameters and we have proved that they are valid \ac{CSS} codes using chain complexes from topology. 
Additionally, by introducing $d-1$ twists in the construction of cylindrical codes, we have obtained a quasi-planar \ac{CSS} codes, named M\"{o}bius codes. 
Starting from quantum MacWilliams identities, we have exploited the undetectable errors \acl{WE} to study analytically the error correction capability of cylindrical and M\"{o}bius codes. 
In particular, although such codes are tailored for asymmetric errors, they outperform standard surface codes both on depolarizing and polarizing channels. 
Finally, we have provided a numerical analysis of the performance of symmetric and asymmetric cylindrical and M\"{o}bius codes using a \ac{MWPM} decoding. 
The investigation shows that the proposed codes are particularly well-suited for handling asymmetric errors.


\bibliographystyle{IEEEtran}
\bibliography{Files/IEEEabrv,Files/StringDefinitions,Files/StringDefinitions2,Files/refs}

\end{document}